\documentclass[12pt, draftclsnofoot, onecolumn]{IEEEtran}
\usepackage{amssymb}
\usepackage{amsthm}
\usepackage{mathtools, cuted}
\usepackage[noabbrev]{cleveref}
\usepackage{caption}
\usepackage{subcaption}
\usepackage[noadjust]{cite}
\usepackage{filecontents}
\usepackage{graphicx}

\DeclarePairedDelimiter\abs{\lvert}{\rvert}%
\DeclarePairedDelimiter\norm{\lVert}{\rVert}%

\newcommand{\coherenceSymbols}{\tau_c}
\newcommand{\ulCoherenceSymbols}{C_u}
\newcommand{\dlCoherenceSymbols}{C_d}
\newcommand{\ulSnr}{\mu}
\newcommand{\dlSnr}{\lambda}

\newcommand{\lmmseChannel}{\hat{\mathbf{h}}}
\newcommand{\leastSqrChannel}{\hat{\mathbf{h}}^{LS}}
\newcommand{\eWiseChannel}{\hat{\mathbf{h}}^{\mathrm{el}}}

\newcommand{\yForLmmse}{\mathbf{Y}^{(p)}_{j}}
\newcommand{\pilotForLmmse}{\mathbf{p}}

\newcommand{\yForRhat}{\mathbf{Y}^{(r)}_{j}}
\newcommand{\pilotForRhat}{\boldsymbol{\phi}}

\newcommand{\rawRestimate}{\ddot{\mathbf{R}}}
\newcommand{\regRestimate}{\hat{\mathbf{R}}}
\newcommand{\elementwiseR}{\mathbf{S}}
\newcommand{\elementwiseRestimate}{\ddot{\elementwiseR}}
\newcommand{\regElementwiseRestimate}{\hat{\elementwiseR}}
\newcommand{\pilotLengthForR}{N_R}
\newcommand{\biasFactorForR}{\alpha_R}
\newcommand{\regRactual}{\bar{\mathbf{R}}}
\newcommand{\regElementwiseRActual}{\bar{\elementwiseR}}
\newcommand{\biasRMatrix}{\mathbf{R}_b}
\newcommand{\pilotLengthThreshold}{\bar{N}_R}

\newcommand{\rawQestimate}{\hat{\mathbf{Q}}}

\newcommand{\elementwiseQ}{\mathbf{P}}
\newcommand{\elementwiseQestimate}{\hat{\elementwiseQ}}
\newcommand{\pilotLengthForQ}{N_Q}
\newcommand{\biasFactorForQ}{\alpha_Q}
\newcommand{\regElementwiseQestimate}{\hat{\elementwiseQ}}
\newcommand{\unRegElementwiseQestimate}{\ddot{\elementwiseQ}}
\newcommand{\biasQMatrix}{\elementwiseQ_b}
\newcommand{\expectationMatrixFirstOrder}{\mathbf{E}}
\newcommand{\expectationMatrixSecondOrder}{\mathbf{G}}

\newcommand{\contamination}{\mathbf{g}}
\newcommand{\singleObservationR}{\breve{\mathbf{R}}}
\newcommand{\elementWiseSingleObservationR}{\breve{\elementwiseR}}

\newcommand{\randPhase}{\theta}
\newcommand{\Rsum}{\mathbf{R}_{s}}
\newcommand{\dlRsum}{\mathbf{R}^{(dl)}_{s}}
\newcommand{\elementWiseRsum}{\mathbf{S}_{s}}
\newcommand{\arbSymMatrix}{\mathbf{C}}
\newcommand{\arbSqrMatrix}{\mathbf{A}}
\newcommand{\arbDiaMatrix}{\mathbf{D}}
\newcommand{\wishartQ}{\tilde{\mathbf{Q}}}
\newcommand{\elementWiseWishartQ}{\tilde{\elementwiseQ}}
\newcommand{\nCubeConstant}{\kappa_1}
\newcommand{\nSqareConstant}{\kappa_2}
\newcommand{\elementWiseNSqrConst}{\kappa_3}
\newcommand{\elementWiseNCubeConst}{\kappa_4}
\newcommand{\genericConstantOne}{\tau}
\newcommand{\arbWishart}{\mathbf{X}}
\newcommand{\arbDiagonalWishart}{\mathbf{Y}}

\newcommand{\genericWEst}{\hat{\mathbf{W}}}
\newcommand{\unRegWEst}{\hat{\mathbf{W}}}

\newcommand{\elementwiseWest}{\hat{\mathbf{W}}}

\newcommand{\unRegWActual}{\bar{\mathbf{W}}}
\newcommand{\regWActual}{\bar{\mathbf{W}}}
\newcommand{\elementwiseWActual}{\bar{\mathbf{W}}}

\newcommand{\Rsumgen}{\bar{\mathbf{R}}_{s}}
\newcommand{\Ssumgen}{\bar{\mathbf{S}}_{s}}

\newcommand{\trace}{\mathrm{tr}}
\newcommand{\diag}{\mathrm{diag}}
\DeclareMathOperator*{\argmin}{\arg\min}

\ifCLASSINFOpdf
\else
\fi
\usepackage{amsmath}
\newtheorem{lemmas}{Lemma}
\newtheorem{theorems}{Theorem}
\newtheorem{remark}{Remark}

\begin{document}
	\title{On the Spectral Efficiency of Massive MIMO Systems with Imperfect Spatial Covariance Information}
	
	%
	
	\author{Atchutaram~K.~Kocharlakota,~\IEEEmembership{Student Member,~IEEE,}
		Karthik~Upadhya,~\IEEEmembership{Member,~IEEE,}
		and~Sergiy~A.~Vorobyov,~\IEEEmembership{Fellow,~IEEE}.
		\thanks{The first and last authors are and the second author was with  the  Department  of  Signal  Processing  and  Acoustics, Aalto  University,  FI-00076 Aalto, Finland
			(e-mail: kameswara.kocharlakota@aalto.fi; karthik.upadhya@gmail.com; svor@ieee.org)}
	}
	
	
	\maketitle
	
	\begin{abstract}
		This paper studies the impact of imperfect channel covariance information on the uplink (UL) and downlink (DL) spectral efficiencies (SEs) of a time-division duplexed (TDD) massive multiple-input multiple-output (MIMO) system. We derive closed-form expressions for the UL and DL average SEs by considering linear minimum mean squared (LMMSE)-type  and element-wise LMMSE-type channel estimation that represent LMMSE and element-wise LMMSE with estimated covariance matrices, respectively. The closed-form expressions of these average SEs are functions of the number of observations used for estimating the spatial covariance matrices of individual and contaminated channels of a target user, and thus enable us to select these key parameters to achieve the desired SE. We present a theoretical analysis of SE behavior for different values of these parameters, followed by simulations, which also demonstrate and validate this behavior. Specifically, we present the SEs computed using estimated covariance matrices and show the accurate agreement between the theoretical and simulated SEs as functions of the number of observations for estimating the covariance matrices of individual and contaminated channels of a user. We also compare these SEs across channel estimation techniques using analytical and simulation studies.
	\end{abstract}
	

	\IEEEpeerreviewmaketitle

	\section{Introduction}
	\IEEEPARstart{A } multi-user massive multiple-input multiple-output (MIMO) system comprises multiple cells, each having a base station (BS) with a large number of antennas (hundreds) to serve multiple users (tens) within the cell. It is considered to be one of the key technologies for the fifth-generation (5G) cellular systems due to the considerable improvement in spectral efficiency (SE) through spatial multiplexing \cite{5595728,6736761,6798744,6824752,6375940} achieved with low computational complexity \cite{5595728,6415389,6415388}. However, acquiring channel state information (CSI) at the base station (BS) is essential to realize the benefits of a massive MIMO system.
	
	In a time-division duplexing (TDD) massive MIMO system, CSI is acquired through uplink (UL) pilots. In time-variant channels, the channels in two different coherence blocks, which is a collection of symbols within a coherence time and bandwidth, are uncorrelated. Consequently, the channel has to be estimated in each coherence block. The number of orthogonal pilots available for channel estimation in a coherence block is limited by the number of available symbols in the coherence block that are not reserved for UL data and DL data, and as a result, UL pilot sequences need to be reused by users across the cells, causing the pilot contamination problem \cite{5595728, 5898372,7996671}. 
	
	With a matched filter receiver combiner, the interference caused to a target user by the users sharing the same pilot is shown to impose a ceiling on the throughput \cite{5595728} as the number of antennas at the BS grows to infinity. This ceiling is due to both the coherent beamforming gain as well as the coherent interference from pilot contamination that increases proportionately with the number of antennas. Several pilot decontamination techniques have been studied to overcome this problem \cite{5898372,7294693,6756975,6415397,7472301,7865983}.
	
	Despite the presence of pilot contamination, under the assumption that the covariance matrices of interfering users are asymptotically linearly independent to each other, the sum rate of the massive MIMO system has been recently proven to be unbounded \cite{8094949}. However, the authors assume that contamination-free covariance matrices of individual users are available at the BS, while, in practice, these covariance matrices also have to be estimated at the BS. Therefore, it is essential to study the performance of a massive MIMO system for a more realistic case where the covariance matrices are estimated. Nonetheless, covariance matrix estimation is a non-trivial task because the channel estimates from which the covariance matrix estimates are obtained are themselves contaminated. Naively utilizing the contaminated channel estimates in a sample covariance estimator will result in the target user covariance matrix estimate containing the covariance matrices of the interference users.
	
	Methods for estimating the individual covariance matrices in the presence of pilot contamination have been recently studied in \cite{caire2017massive,neumann2018covariance,bjornson2016imperfect,upadhya2018covariance}. In all these works, the authors assume that the channel covariance matrices are constant across multiple coherence blocks, and then, the observations from a few of these coherence blocks are used to estimate the covariance matrices. In \cite{caire2017massive}, the authors first estimate the angle-delay power spread function from the contaminated channel estimates of multiple coherence blocks, then use this function for supervised/unsupervised clustering of the multipath components belonging to the target user. Finally, they use the clusters to estimate the spatial covariance matrix of the target user. In \cite{neumann2018covariance}, the authors develop a method where the pilot allocation is changed in each coherence block. The channel estimates obtained from these blocks are then used to obtain a maximum-likelihood estimate of the contamination-free covariance matrix. Work \cite{bjornson2016imperfect} presents two methods which avoid contamination in the covariance matrices by utilizing dedicated orthogonal pilots for each user for estimating its individual spatial covariance matrix. In \cite{upadhya2018covariance}, a new pilot structure and a covariance matrix estimation method are developed that offer higher throughput and lower mean squared error (MSE) of the channel estimates than earlier methods. Although \cite{upadhya2018covariance} requires additional pilots for estimating the individual covariance matrices of each user, it does not assume any additional structures on the covariance matrices of the users, and it does not require backhaul communication between the neighboring cells, unlike \cite{caire2017massive} and \cite{neumann2018covariance}, respectively. Moreover, since the additional pilots in \cite{upadhya2018covariance} are not dedicated to each user as in \cite{bjornson2016imperfect}, the number of additional pilots in \cite{upadhya2018covariance} does not grow with the total number of users in the entire system. Therefore, in this work, we consider the covariance estimation method of \cite{upadhya2018covariance} to study the performance in the massive MIMO system.
	
	Utilizing the estimated covariance matrices for channel estimation results in a trade-off in the SE, since increasing the number of additional pilots to estimate the covariance matrices will not only improve the quality of the covariance estimate (and hence, the channel estimate) but also increase the estimation overhead. Consequently, the number of additional pilots for estimating the covariance matrices becomes a key trade-off parameter for the optimal performance. Therefore, closed-form expressions that relate the SE for UL and DL channels with the number of additional covariance pilots prove to be of key importance. Closed-form expressions for UL ergodic achievable SE in single and multi-cell massive MIMO systems with various linear receive-combiners designed using the minimum mean-squared error (MMSE) channel estimate have been derived in \cite{ngo2013energy} and \cite{ngo2013multicell}, respectively. Similar expressions for the achievable SE in the DL have been derived in \cite{5898372}. However, the closed-form expressions in the aforementioned articles have been derived under the assumption of imperfect CSI and perfect covariance information. Closed-form expressions for the spectral efficiency expressions for the case of estimated covariance matrices have not yet been derived, to the best of our knowledge.
	
	In this paper, we derive closed-form expressions for the average UL, and DL SEs in a massive MIMO system with LMMSE-type/element-wise LMMSE-type channel estimation that uses estimated covariance matrices, obtained using the method in  \cite{upadhya2018covariance}, in LMMSE/element-wise LMMSE channel estimate \footnote{Some preliminary results are also reported in \cite{UsICASSP19}.}. Note that, in this paper, we use LMMSE-type/element-wise LMMSE-type to denote the channel estimation with estimated covariance matrices, and use LMMSE/element-wise LMMSE to denote channel estimation with true covariance matrices.

	The following are the contributions of this paper.
	\begin{itemize}
		\item We derive closed-form expressions for the average UL and DL spectral efficiencies when the LMMSE-type and element-wise LMMSE-type channel estimates are used in a matched filter combiner.
		\item We also derive expressions for the average UL and DL SE when the regularized covariance matrix estimates are used in the element-wise LMMSE-type channel estimates.
		\item Using theoretical and simulation studies on the derived SE expressions, we find out and demonstrate that the number of additional pilots needed for covariance estimation as a key trade-off parameter.
		\item We compare the performance of the element-wise LMMSE-type channel estimate with the LMMSE-type channel estimate. To the best of our knowledge, this is the first work that quantitatively compares the average UL/DL SE obtained with LMMSE-type and element-wise LMMSE-type estimates.
	\end{itemize}
	
	The paper is organized as follows. In Section~II, we describe the system model along with a detailed explanation on the channel estimation and covariance matrices estimation methods. Section~III reports our main derivations in order to obtain closed-form expressions for the UL and DL SEs for three different combinations of channel estimation techniques. We present a detailed theoretical discussion on the derived closed-form expressions in Section IV, where we analyze the behavior of SE as a function of pilot overhead for covariance estimation. Section~V provides the simulation results and their comparison with the main results obtained in Section~III. We conclude this work in Section~VI. Technical proofs of lemmas and theorems in the paper appear in appendices at the end of the paper.
	
	\textit{Notation}: We use boldface capital letters for matrices, and boldface lowercase letters for vectors. The superscripts $(\cdot)^*$, $(\cdot)^{\intercal}$, and $(\cdot)^H$ denote element-wise conjugate, transpose, and Hermitian transpose operations, respectively. Moreover, $\mathcal{CN}(\mathbf{m}, \mathbf{R})$ denotes (circularly symmetric) complex Gaussian random vector with mean vector $\mathbf{m}$ and covariance matrix $\mathbf{R}$, while $\mathcal{W}(N, \mathbf{R})$ denotes Wishart random matrix with $N$ degrees of freedom and $\mathbf{R}$ is the covariance matrix that corresponds to underlying Gaussian random vectors. In addition, $\mathcal{U}[x_1,x_2]$ stands for the uniform distribution between $x_1$ and $x_2$. The element in $i^{th}$ row and $j^{th}$ column of the matrix $\mathbf{A}$ is denoted as $[\mathbf{A}]_{ij}$, $\mathbf{I}$ stands for an identity matrix (of appropriate size), $\diag(\mathbf{A})$ is a diagonal matrix whose diagonal elements are same as the diagonal elements of the matrix $\mathbf{A}$. We use $\trace(\cdot)$ to denote trace of a matrix, $\norm{\cdot}$ to denote $l_2$ norm of a vector or a matrix, i.e., Frobenius norm, and $\mathbb{E}\{\cdot\}$ stands for the mathematical expectation. Finally, the symbol $\delta_{ij}$ is the Kronecker delta such that $\delta_{ij}$ = 1 if $i = j$, and 0 otherwise. 
	\vspace{-6mm}
	
	\section{System Model}
	We consider a massive MIMO system with $L$ cells, each having a BS with $M$ antennas and serving $K$ single-antenna users. The channel between user $(l, k)$ ($k^{th}$ user in $l^{th}$ cell) and BS $j$ is denoted as $\mathbf{h}_{jlk} \in \mathbb{C}^M$ and is assumed to be distributed as $\mathcal{CN}(\mathbf{0},\mathbf{R}_{jlk})$, where ${\mathbf{R}_{jlk}\triangleq \mathbb{E}\{\mathbf{h}_{jlk}\mathbf{h}_{jlk}^H\}}$ is the spatial covariance matrix. We assume the block-fading model where the channel is assumed to be constant over the coherence bandwidth $B_c$ and coherence time $T_c$. In other words, the channel is assumed to be constant over a coherence block containing $\coherenceSymbols = B_c T_c$ symbols. 
	
	We consider TDD transmission and each coherence block is divided into slots for UL pilots, UL and DL data. The number of data symbols in the UL and DL time slot is denoted as $\ulCoherenceSymbols$ and $\dlCoherenceSymbols$, respectively. The channel is assumed to be reciprocal, i.e., the DL channel between BS $j$ and user $(l, k)$ can be written as $\mathbf{h}^*_{jlk}$, and consequently, the channel estimated in the UL is used in designing the DL precoding matrix. This is represented in Fig.~\ref{fig:pilotpositioning}(\subref{fig:coherenceblock_ChEst}).
	
	We consider two types of UL pilots, namely, (i) pilots for estimating the channel (also referred to as ChEst pilots) and (ii) pilots for estimating the covariance matrix (referred to as CovEst pilots). Both ChEst pilots and CovEst pilots are assumed to be of length $P$ symbols.
	
	The spatial covariance matrices are assumed to be constant over a considerably longer time-interval and bandwidth than a single coherence block \cite{caire2017massive,bjornson2016imperfect,neumann2018covariance,upadhya2018covariance}. Specifically, we assume that the covariance matrices are coherent over the time-interval $T_s$ and system bandwidth $B_s$, which implies that they can be assumed to be constant over $\tau_s = B_sT_s/B_cT_c = B_sT_s/\tau_c$ coherence blocks (usually several tens of thousands of blocks in practice). This time-frequency grid over which the second-order statistics of the channel are assumed to be constant is illustrated in Fig.~\ref{fig:pilotpositioning}(\subref{fig:fullgrid}).
	
	Each of the $\tau_s$ coherence blocks contain ChEst pilots for channel estimation, whereas only $N_R$ out of the $\tau_s$ coherence blocks contain CovEst pilots in addition to the ChEst pilots (as can be seen in Fig.~\ref{fig:pilotpositioning}(\subref{fig:fullgrid})). The coherence blocks that contain the CovEst pilots are depicted in Fig.~\ref{fig:pilotpositioning}(\subref{fig:coherenceblock_CovEst}).
	
	\begin{figure}
		\begin{subfigure}{.3\textwidth}
			\centering
			\includegraphics[width=.8\linewidth,trim={0cm 0cm 0cm 2cm}]{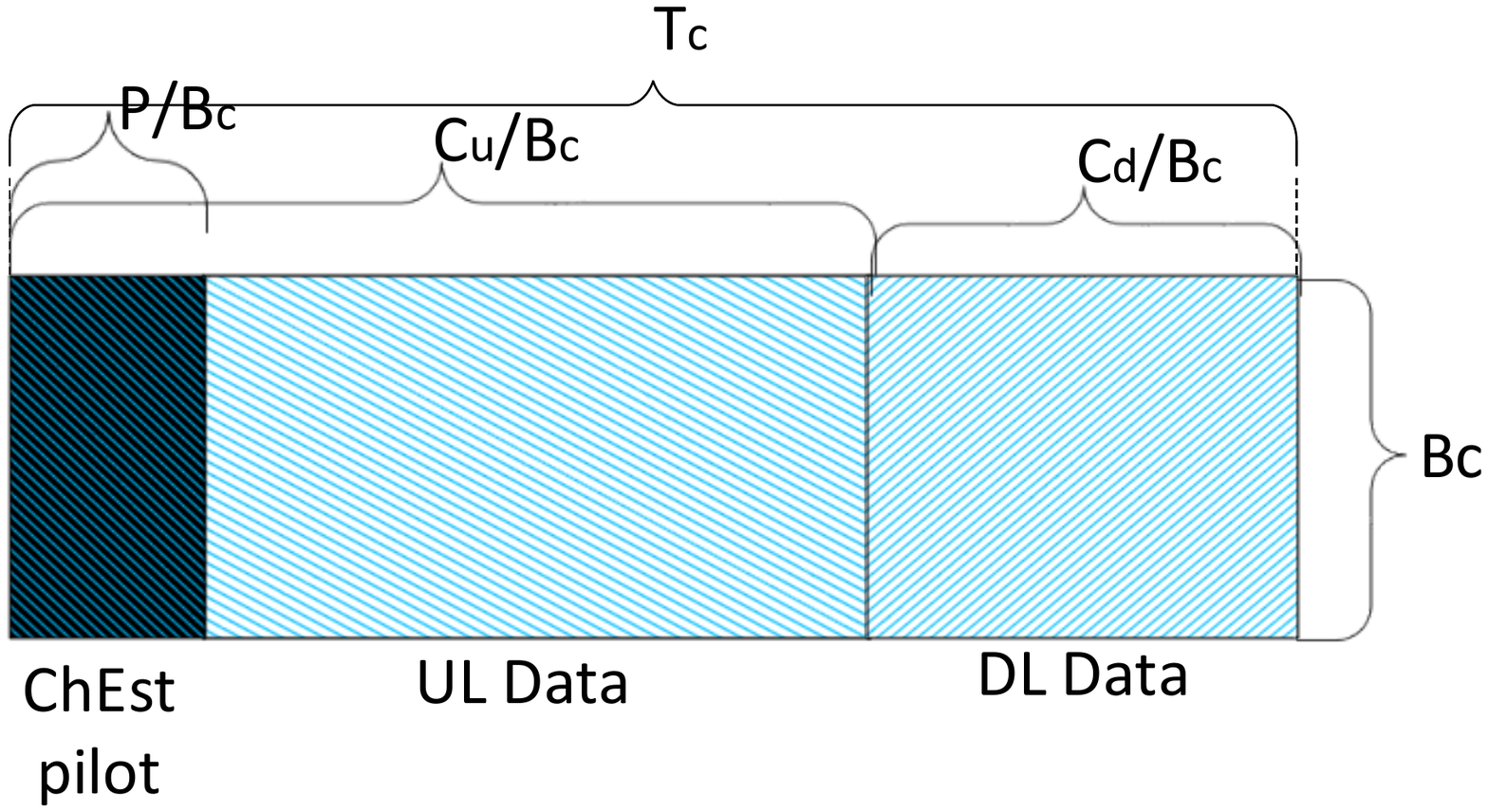}
			\caption{Coherence block with only ChEst pilots.}
			\label{fig:coherenceblock_ChEst}
		\end{subfigure}
		\hspace{0.3cm}
		\begin{subfigure}{.3\textwidth}
			\centering
			\includegraphics[width=.8\linewidth,trim={0cm 0cm 0cm 2cm}]{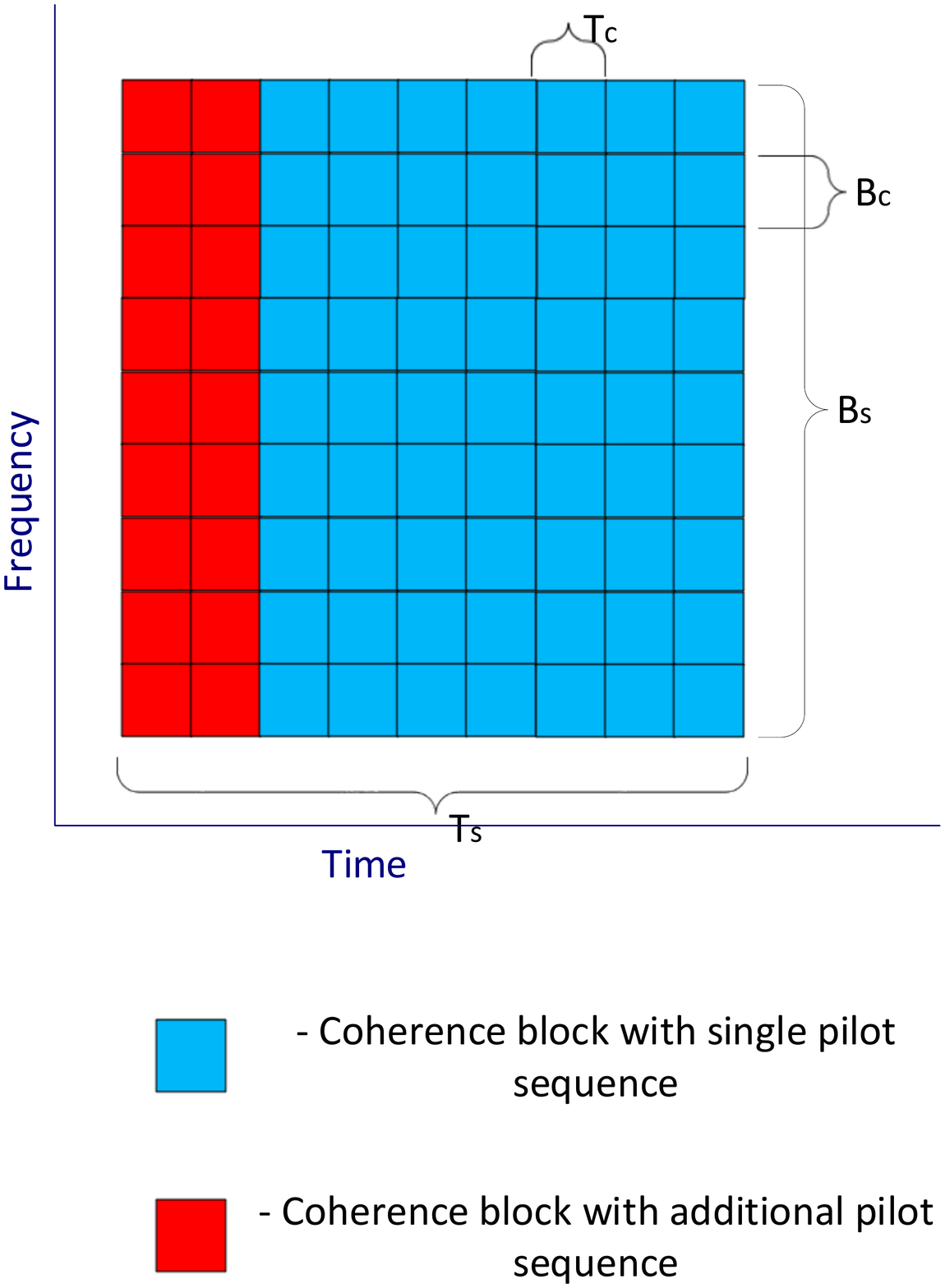}
			\caption{Grid of coherence blocks with coherent covariance matrices.}
			\label{fig:fullgrid}
		\end{subfigure}%
		\hspace{0.3cm}
		\begin{subfigure}{.3\textwidth}
			\centering
			\includegraphics[width=.8\linewidth,trim={0cm 0cm 0cm 2cm}]{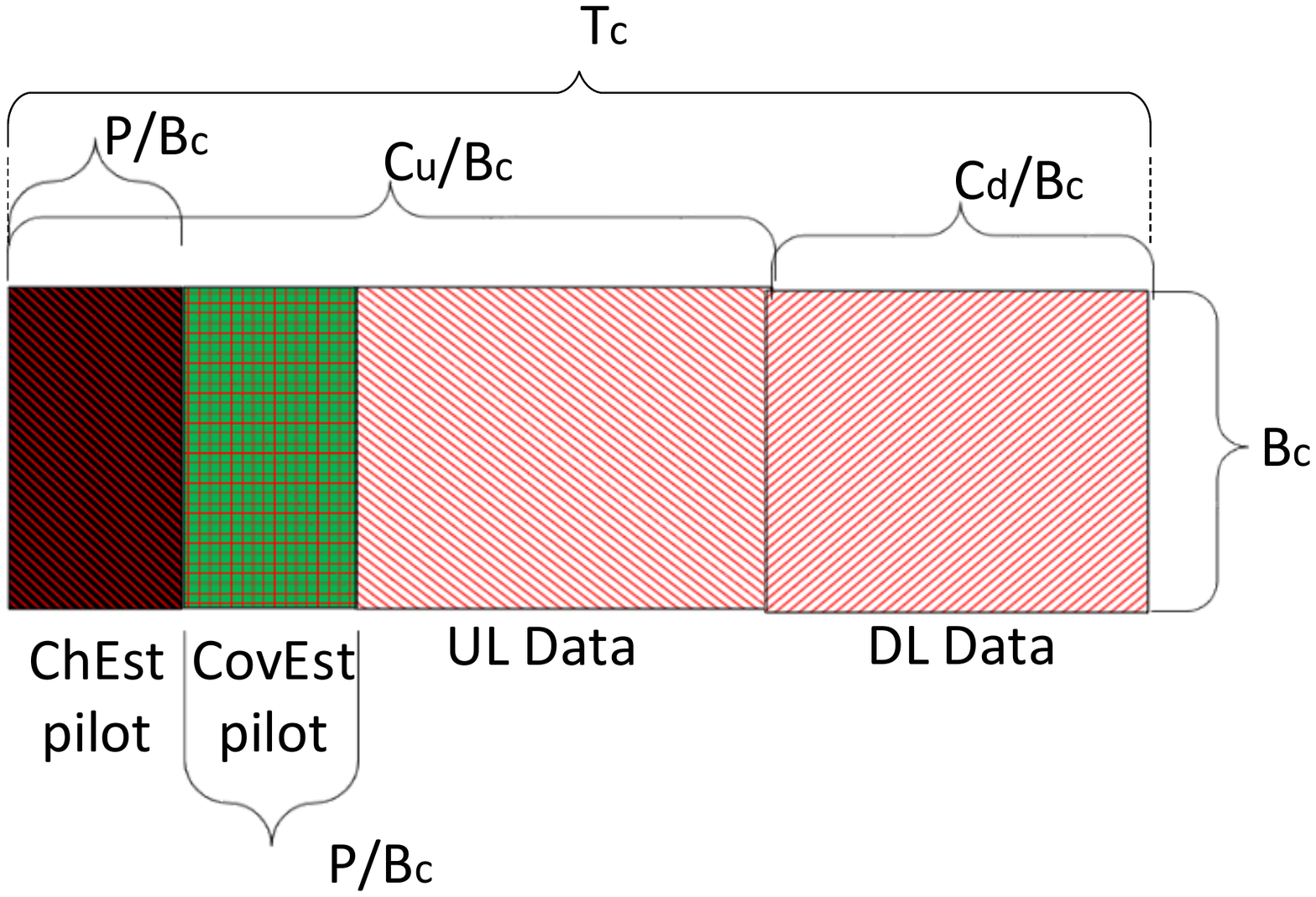}
			\caption{Coherence block with additional CovEst pilots.}
			\label{fig:coherenceblock_CovEst}
		\end{subfigure}
		\caption{Time frequency grid and pilot positioning. \vspace{-6mm}}
		\label{fig:pilotpositioning}
	\end{figure}
	
	The UL received signal, $\mathbf{Y}_j[n] \in \mathbb{C}^{M \times \ulCoherenceSymbols}$, in the $n^{th}$ coherence block at BS $j$ can be written as \vspace{-3mm}
	\begin{align}
	\mathbf{Y}_j[n] = \sum_{l=1}^{L}\sum_{k=1}^{K} \sqrt{\ulSnr}\mathbf{h}_{jlk} \mathbf{x}^{\intercal}_{lk}[n] + \mathbf{N}_j[n]
	\label{eqn:receivedSignalData}
	\end{align}
	where $\mathbf{x}_{lk} \in \mathbb{C}^{\ulCoherenceSymbols}$ is the signal transmitted by user $(l, k)$, $\mathbf{N}_j \in \mathbb{C}^{M\times \ulCoherenceSymbols}$ is the additive white Gaussian noise at the BS, and $\ulSnr$ is the UL transmit power. The transmitted data $ \mathbf{x}_{lk} $ is assumed to be distributed as $\mathbf{x}_{lk} \sim \mathcal{CN}(\mathbf{0},\mathbf{I})$ whereas the elements of $\mathbf{N}_j$ are assumed to be identically and independently distributed (i.i.d) as $ \mathcal{CN}(0,1)$.
	
	In the DL, the received signal $\mathbf{z}_{ju}[n] \in \mathbb{C}^{\dlCoherenceSymbols}$ at user $(j,u)$ in the $n^{th}$ coherence block can be written as \vspace{-3mm}
	\begin{align*}
	\mathbf{z}_{ju}[n] = \sum_{l=1}^{L} \sum_{k=1}^{K} \sqrt{\dlSnr} (\mathbf{h}^H_{jlu} \mathbf{b}_{lk} )\mathbf{d}_{lk}[n] + \mathbf{e}[n]
	\end{align*}
	where $\mathbf{d}_{lk} \in \mathbb{C}^{\dlCoherenceSymbols}$ is the payload data from BS $l$ to its user $(l,k)$, $\mathbf{b}_{lk} \in \mathbb{C}^{M}$ is the corresponding precoding vector normalized such that the average transmitted power is $\lambda$, i.e., $\mathbb{E}\{\norm{\mathbf{b}_{lk}}^2\}= 1$, and ${\mathbf{e} \in \mathbb{C}^{\dlCoherenceSymbols}}$ is the additive white Gaussian noise distributed as $\mathcal{CN}(\mathbf{0}, \mathbf{I})$.
	
	
	In the following subsections, we explain the pilot structure in detail and describe the methods used for channel and covariance matrix estimation. \vspace{-3mm}
	\subsection{Channel Estimation}
	A dedicated set of $P$ ($\geq K$) symbols is allocated to UL pilots for channel estimation in each coherence block, as shown in Figs.~\ref{fig:pilotpositioning}(\subref{fig:coherenceblock_ChEst}) and \ref{fig:pilotpositioning}(\subref{fig:coherenceblock_CovEst}). In other words, let $\pilotForLmmse_k \in \mathbb{C}^{P}$ denote the ChEst pilot sequence used by the $k^{th}$ user in any of the $L$ cells. Then, for another user $m$ in the same cell, we have $\pilotForLmmse^{H}_k \pilotForLmmse_m = P\delta_{km}$. We assume a pilot-reuse factor of $1$, implying that the same $P$ pilots are reused in each cell and each user is randomly allocated one of these pilots for channel estimation \footnote{Deriving the results in Section \ref{sec:mainResults} for arbitrary pilot-reuse factors greater than $1$ is straightforward.}.
	
	The pilot transmissions in all cells are assumed to be synchronized. Then, the received signal at BS $j$ during pilot transmissions in the $n^{th}$ coherence block (denoted as $\yForLmmse[n]$) can be written as \vspace{-3mm}
	\begin{align}
	\yForLmmse[n] = \sum_{l=1}^{L}\sum_{k=1}^{K} \sqrt{\ulSnr}\mathbf{h}_{jlk} \pilotForLmmse^{\intercal}_k + \mathbf{N}^{(p)}_j[n] \label{Yp}
	\end{align}
	where $\mathbf{N}^{(p)}_j[n] \in \mathbb{C}^{M\times P}$ is the noise during pilot transmission. 
	
	We consider LMMSE and element-wise LMMSE techniques for estimating $\mathbf{h}_{jlk}$ from the observed signal $\yForLmmse$ given in \eqref{Yp}. In what follows, we first discuss these estimation techniques when the channel covariance information is available at the BS, and subsequently, we discuss the practical case where this information is estimated at the BS.
	\subsubsection{LMMSE Channel Estimation}
	\label{subsubsec:lmmse}
	From \eqref{Yp}, the least-squares (LS) channel estimate of user $(j,u)$ at BS $j$ in the $n^{th}$ coherent block (denoted as $\leastSqrChannel_{jju}[n]$) can be obtained by solving the optimization problem \vspace{-3mm}
	\begin{align*}
	\leastSqrChannel_{jju}[n] &= \argmin_{\mathbf{g}} \quad \norm{\yForLmmse[n] - \sqrt{\ulSnr}\mathbf{g}\pilotForLmmse^{\intercal}_{u}}^2
	\end{align*}
	the solution of which is given by \vspace{-3mm}
	\begin{align*}
	&\leastSqrChannel_{jju}[n] = \frac{1}{P\sqrt{\ulSnr}}\yForLmmse[n]\pilotForLmmse^{*}_{u}
	= \mathbf{h}_{jju}+\sum_{l\ne j} \mathbf{h}_{jlu} + \frac{1}{P\sqrt{\mu}}\mathbf{N}^{(p)}_j[n]\pilotForLmmse^{*}_{u}.
	\end{align*}
	As the aforementioned LS channel estimate serves as a sufficient statistic for $\mathbf{h}_{jju}$, the  LMMSE estimate of the channel of a target user $(j, u)$ at BS $j$ in the $n^{th}$ coherent block, $\lmmseChannel_{jju}^{LMMSE}[n] = \mathbf{W}\leastSqrChannel_{jju}[n]$, can be obtained by solving for $\mathbf{W}$ as follows \vspace{-3mm}
	\begin{align*}
	\mathbf{W}&=  \argmin_{\mathbf{G}} \quad \mathbb{E}\{\norm{\mathbf{h}_{jju} - \mathbf{G}\leastSqrChannel_{jju}[n]}^2\}.
	\end{align*}
	Here, the expectation is with respect to the additive noise, and the channel realizations of all the users in the system, in the $n^{th}$ coherent block. Finally, the resultant LMMSE channel estimate is given by \vspace{-3mm}
	\begin{align}
	&\lmmseChannel_{jju}^{LMMSE}[n]= \mathbf{R}_{jju}\mathbf{Q}^{-1}_{ju}\leastSqrChannel_{jju}[n] \label{lmmse}\\
	&\mathbf{Q}_{ju} \triangleq \mathbb{E}\{\leastSqrChannel_{jju}[n](\leastSqrChannel_{jju}[n])^{H}\} = \sum_{l=1}^{L}\mathbf{R}_{jlu} + \frac{1}{P\ulSnr}\mathbf{I} \;. \nonumber
	\end{align}
	
	\subsubsection{Element-wise LMMSE Channel Estimation}
	\label{subsubsec:elementwise-lmmse}
	Obtaining the LMMSE channel estimates involves inverting an $M\times M$ matrix, which is computationally expensive when $M$ is large. An alternative approach is to use the element-wise LMMSE estimate in which the correlation between channel coefficients across the antennas is neglected and only the diagonal elements of the covariance matrices are considered for channel estimation. This technique has the additional advantage that it requires a fewer number of samples/pilots for the covariance estimation that does not grow with $M$ \cite{8094949}. 
	
	The element-wise LMMSE estimate of the channel can be obtained as
	\begin{align}
	[\hat{\mathbf{h}}^{\mathrm{el-LMMSE}}_{jju}[n]]_{p} &= \frac{[\elementwiseR_{jju}]_{pp}}{[\elementwiseQ_{ju}]_{pp}}[\leastSqrChannel_{jju}[n]]_{p}, \quad p \in \{1,\dots, M\} \label{el_lmmse}
	\end{align}
	where $\elementwiseR_{jju} \triangleq \diag(\mathbf{R}_{jju})$ and $\elementwiseQ_{ju} \triangleq \diag(\mathbf{Q}_{ju})$. Note that the structure in \eqref{el_lmmse} directly follows from the structure \eqref{lmmse} by ignoring the non diagonal elements of the covariance matrices.
	
	\subsubsection{LMMSE-type and Element-wise LMMSE-type Channel Estimation With Imperfect Channel Covariance Matrices}

	Although the channel estimates in Sections \ref{subsubsec:lmmse} and \ref{subsubsec:elementwise-lmmse} assume that the covariance information is known, in practice it has to be estimated at the BS. Therefore, it is reasonable to replace these matrices with estimated covariance matrices. Then the LMMSE-type and element-wise LMMSE-type channel estimates given in \eqref{lmmse} and \eqref{el_lmmse} can be re-written as
	\begin{align}
	\lmmseChannel_{jju}[n] &= \regRestimate_{jju}\rawQestimate^{-1}_{ju}\leastSqrChannel_{jju}[n] \label{estimateCovaraianceLMMSE}\\
	[\eWiseChannel_{jju}[n]]_{p} &= \frac{[\regElementwiseRestimate_{jju}]_{pp}}{[\elementwiseQestimate_{ju}]_{pp}}[\leastSqrChannel_{jju}[n]]_{p}, \quad p \in \{1,\dots, M\} \label{el_estimateCovaraianceLMMSE}
	\end{align}
	where $\regRestimate_{jju}$, $\rawQestimate_{ju}$, $\regElementwiseRestimate_{jju}$, and $\elementwiseQestimate_{ju}$ are estimates of $\mathbf{R}_{jju}$, $\mathbf{Q}_{ju}$, $\elementwiseR_{jju}$, and $\elementwiseQ_{ju}$, respectively. 
	
	Note that while the channel estimates in \eqref{estimateCovaraianceLMMSE} and \eqref{el_estimateCovaraianceLMMSE} have the same structure as in \eqref{lmmse} and \eqref{el_lmmse}, they are formally not the LMMSE and element-wise LMMSE channel estimates because of the fact that they utilize the estimated covariance matrices. Consequently, we have removed the LMMSE superscript in \eqref{estimateCovaraianceLMMSE} and \eqref{el_estimateCovaraianceLMMSE} to make this distinction.
	
	In Section \ref{sec:mainResults}, we use \eqref{estimateCovaraianceLMMSE} and \eqref{el_estimateCovaraianceLMMSE} as the channel estimates. The following subsection is dedicated to describing the pilot structures and the techniques for estimating these matrices. \vspace{-3mm}
	
	\subsection{Covariance Matrix Estimation}
	In a multi-cell massive MIMO system, since the channel estimates are contaminated, estimating contamination-free spatial covariance matrices of individual users, i.e., $\mathbf{R}_{jlk}$ from these channel estimates is non-trivial. Naively using the channel estimates in a sample covariance estimator will result in the estimate of the covariance matrix of the target user being contaminated by the covariance matrices of users that share the same pilot with the target user. 
	
	Several methods addressing this problem have been proposed in recent literature \cite{caire2017massive,bjornson2016imperfect,neumann2018covariance,upadhya2018covariance}. However, among these methods, only the estimators in \cite{bjornson2016imperfect} and \cite{upadhya2018covariance} are in closed-form and consequently, lend themselves to analysis. Moreover, since \cite{upadhya2018covariance} is seen to outperform \cite{bjornson2016imperfect}, we select the estimator in \cite{upadhya2018covariance} for performance analysis when the estimate is used for both LMMSE-type and element-wise LMMSE-type channel estimation. 
	
	In this subsection, we briefly describe the pilot structure introduced in \cite{upadhya2018covariance} and the corresponding spatial covariance estimation method. 
	The objective is to compute a pair of $\regRestimate_{jju}$ and $\rawQestimate_{ju}$ (or $\regElementwiseRestimate_{jju}$ and $\elementwiseQestimate_{ju}$) for each set of $\tau_s$ contiguous coherence blocks (over which the second-order channel statistics can be assumed constant). 
	
	To obtain $\rawQestimate_{ju}$, since the matrix $\mathbf{Q}_{ju}$ is defined as the covariance matrix of the LS channel estimate $\lmmseChannel^{LS}_{jju}[n]$, we use the LS channel estimates from multiple coherence blocks in a sample covariance estimator. As described in Subsection II-A, these LS channel estimates are obtained from the ChEst pilot sequence $\pilotForLmmse_{k}$ that is transmitted by the $k^{th}$ user in all the cells in each coherence block (Fig.~\ref{fig:pilotpositioning}). We use a set of $\pilotLengthForQ$ ($\geq M$) number of LS estimates out of the available $\tau_s$ number of LS channel estimates for computing $\rawQestimate_{ju}$. Therefore, we have \vspace{-3mm}
	\begin{align*}
		\rawQestimate_{ju} = \frac{1}{\pilotLengthForQ}\sum_{n=1}^{\pilotLengthForQ}\lmmseChannel^{LS}_{jju}[n](\lmmseChannel^{LS}_{jju}[n])^{H}
	\end{align*}
	and $\mathbb{E}\{\rawQestimate_{ju}\} = \mathbf{Q}_{ju}$. Similarly, the unbiased estimate of $\elementwiseQ_{ju}$ is obtained using a sample covariance estimator as follows \vspace{-3mm}
	\begin{align*}
	[\elementwiseQestimate_{ju}]_{pp} = \frac{1}{\pilotLengthForQ}\sum_{n=1}^{\pilotLengthForQ}|[\lmmseChannel^{LS}_{jju}[n]]_p|^2, \quad {\forall p \in {1 \dots M}}.
	\end{align*}
	
	
	
	According to \cite{upadhya2018covariance}, to estimate $\regRestimate_{jju}$, each user transmits an additional pilot sequence of length $P$ symbols for $\pilotLengthForR$ out of the $\tau_s$ coherence blocks (represented as the red coherence blocks in Fig.~\ref{fig:pilotpositioning}). Specifically, the CovEst pilots, denoted as $\{\pilotForRhat_{lk}[n]\}_{n=1}^{\pilotLengthForR}$, are transmitted by the user $(l,k)$, with the pilot sequence in $n^{th}$ coherence block given as a phase-shifted version of the ChEst pilot, i.e., $\pilotForRhat_{lk}[n] = e^{j\randPhase_{ln}}\pilotForLmmse_{k}$. The phase-shifts $\{\randPhase_{ln}\}_{n=1}^{\pilotLengthForR}$ are (pseudo)-random and are generated such that $\{\randPhase_{ln}\}_{n=1}^{\pilotLengthForR}$ is independent of the channel vectors and satisfies $\mathbb{E}\{e^{j\randPhase_{ln}}\} = 0$ \cite{upadhya2018covariance}. A random sequence that satisfies both these properties is $\randPhase_{ln}\sim\mathcal{U}[0,2\pi]$. Furthermore, the random phase sequences are assumed to be i.i.d across cells.
	\vspace{-3mm}
	
	\begin{remark}
		In practice, the phase sequences $\{\randPhase_{ln}\}_{n=1}^{\pilotLengthForR}$ can be obtained using a pseudo-random sequence generator. Each user can then be assigned a sequence based on the cell to which it is associated. 
		
		We also assume that the BSs have knowledge of these sequences, which, in practice, can be accomplished by one of the following two approaches.
		\begin{itemize}
			\item The $L$ sequences are generated before-hand and stored at the user. The BS only conveys its index $\ell$ during initial access.
			\item The BS conveys the seed for the pseudo-random number generator during initial access.
		\end{itemize}
	\end{remark}
	
	Now, let $\yForRhat[n]$ be the received signal when the users transmit the CovEst pilots $\pilotForRhat_{ju}[n]$. Then, $\yForRhat[n]$ can be written as \vspace{-3mm}
	\begin{align}
	\yForRhat[n] = \sum_{l=1}^{L}\sum_{k=1}^{K} \sqrt{\ulSnr}\mathbf{h}_{jlk} \pilotForRhat^{\intercal}_{lk}[n] + \mathbf{N}^{(r)}_j[n] \label{Yr}
	\end{align}
	where $\mathbf{N}^{(r)}_j[n]$ is the additive noise at the BS that has the same statistics as $\mathbf{N}^{(p)}_j[n]$.
	Additionally, we denote LS channel estimates obtained from the pilots $\pilotForLmmse_{u}$ and $\pilotForRhat_{ju}$ as $\lmmseChannel^{(1)}_{jju}[n]$ and $\lmmseChannel^{(2)}_{jju}[n]$, respectively. Using \eqref{Yp} and \eqref{Yr}, $\lmmseChannel^{(1)}_{jju}[n]$ and $\lmmseChannel^{(2)}_{jju}[n]$ can be obtained by solving
	\begin{align*}
	\lmmseChannel^{(1)}_{jju}[n] &\triangleq \argmin_{\mathbf{g}} \quad \norm{\yForLmmse[n] - \sqrt{\ulSnr}\mathbf{g}\pilotForLmmse^{\intercal}_{u}}^2\\
	\lmmseChannel^{(2)}_{jju}[n] &\triangleq \argmin_{\mathbf{g}} \quad \norm{\yForRhat[n] - \sqrt{\ulSnr}\mathbf{g}\pilotForRhat^{\intercal}_{ju}[n]}^2.
	\end{align*}
	By using the fact that $\pilotForRhat_{lk}[n] = e^{j\randPhase_{ln}}\pilotForLmmse_{k}$, the LS estimates are then obtained as
	\begin{align}
	\lmmseChannel^{(1)}_{jju}[n] &= \lmmseChannel_{jju}^{LS}[n] = \mathbf{h}_{jju}+\sum_{l\ne j} \mathbf{h}_{jlu}
	+ \frac{1}{P\sqrt{\mu}}\mathbf{N}^{(p)}_j[n]\pilotForLmmse^{*}_{u} \label{eq: h1}\\
	\lmmseChannel^{(2)}_{jju}[n] &= \frac{1}{P\sqrt{\mu}} 
	\yForRhat[n]\pilotForRhat^{*}_{ju} = \frac{1}{P\sqrt{\mu}} 
	\yForRhat[n]e^{-j\randPhase_{jn}}\pilotForLmmse^{*}_{u} \nonumber\\
	&= \mathbf{h}_{jju}+\sum_{l\ne j} \mathbf{h}_{jlu}e^{j(\randPhase_{ln}-\randPhase_{jn})}
	+ \frac{1}{P\sqrt{\mu}}\mathbf{N}^{(r)}_{j}[n]\pilotForLmmse^{*}_{u}e^{-j\randPhase_{jn}}. \label{eq: h2}
	\end{align}
	
	In the following subsections, we describe both cases of complete and diagonal matrix estimation using the aforementioned LS channel estimates.
	
	\subsubsection{Estimation of $ \regRestimate_{jju} $}
	Note that the second and third terms in \eqref{eq: h1}, corresponding to the interference and noise, respectively, are independent of the second and third terms in \eqref{eq: h2}. This independence arises from the fact that $\randPhase_{ln}$ ($\sim\mathcal{U}[0,2\pi]$) is independent of $\randPhase_{jn}$ (for all $l \ne j$) and the channel realizations. Consequently, the cross-correlation of $ \lmmseChannel^{(1)}_{jju}[n] $ and $\lmmseChannel^{(2)}_{jju}[n]$, gives
	\begin{align*}
	&\mathbf{R}_{\lmmseChannel^{(1)}\lmmseChannel^{(2)}} = \mathbb{E}\{\lmmseChannel^{(1)}_{jju}[n](\lmmseChannel^{(2)}_{jju}[n])^H\}\\
	&= \mathbb{E}\Bigg\{\!\!\!\left\{\mathbf{h}_{jju}\! + \!\sum_{l\ne j} \!\mathbf{h}_{jlu}\! + \!\frac{1}{P\sqrt{\mu}}\mathbf{N}^{(p)}_j[n]\pilotForLmmse^{*}_{u}\right\}\!\!\left\{\mathbf{h}_{jju}\! + \!\sum_{l\ne j} \!\mathbf{h}_{jlu}e^{j(\randPhase_{ln}-\randPhase_{jn})}\! + \!\frac{1}{P\sqrt{\mu}}\mathbf{N}^{(r)}_{j}[n]\pilotForLmmse^{*}_{u}e^{-j\randPhase_{jn}}\right\}^H\!\Bigg\}\\
	&= \mathbf{R}_{jju}.
	\end{align*}
	Therfore, we can use the following unbiased Hermitian-symmetric sample cross-covariance matrix as an estimate for $\mathbf{R}_{jju}$ \cite{upadhya2018covariance}
	\begin{align}
	\rawRestimate_{jju} = \frac{1}{2\pilotLengthForR} \sum_{n=1}^{\pilotLengthForR} \left(\lmmseChannel^{(1)}_{jju}[n]\left(\lmmseChannel^{(2)}_{jju}[n]\right)^{H} + \lmmseChannel^{(2)}_{jju}[n]\left(\lmmseChannel^{(1)}_{jju}[n]\right)^{H}\right). \label{rawR}
	\end{align}
	As $\pilotLengthForR \to \infty$, one can show that the estimated covariance matrix converges in probability to the true covariance matrix, i.e., $\rawRestimate_{jju} \overset{P}{\underset{\pilotLengthForR \to \infty}{\longrightarrow}} \mathbf{R}_{jju}$.
	However, the unbiased covariance estimator given in \eqref{rawR} does not guarantee positive diagonal elements for finite $ \pilotLengthForR $. Therefore, we consider a regularized estimate for the covariance matrix given by
	\begin{align}
	\regRestimate_{jju} = \biasFactorForR \rawRestimate_{jju} + (1-\biasFactorForR)\biasRMatrix \label{regR}
	\end{align}
	where $\biasRMatrix$ is an arbitrary symmetric positive definite bias-matrix, and $\biasFactorForR$ is a design parameter. Additionally, it is useful to define $\regRactual_{jju}$ to denote the expected value of $\regRestimate_{jju}$ as 
	\begin{equation*}\regRactual_{jju} \triangleq  \mathbb{E}\{\regRestimate_{jju}\} = \biasFactorForR \mathbf{R}_{jju} + {(1-\biasFactorForR)}\biasRMatrix.
	\end{equation*}
	\subsubsection{Estimation of $ \elementwiseR_{jju} $}
	For element-wise LMMSE-type estimation, it is sufficient to estimate the diagonal matrices $\elementwiseR_{jju}$ and $\elementwiseQ_{ju}$. Therefore, we present an unbiased Hermitian-symmetric covariance estimate $\elementwiseRestimate_{jju}$ (similar to $\rawRestimate_{jju}$) as follows
	\begin{align}
	[\elementwiseRestimate_{jju}]_{pp} &= \frac{1}{2\pilotLengthForR}\sum_{n=1}^{\pilotLengthForR}[\lmmseChannel^{(1)}_{jju}[n]]_{p}[\lmmseChannel^{(2)}_{jju}[n]]^{*}_p 
	+\frac{1}{2\pilotLengthForR}\sum_{n=1}^{\pilotLengthForR} [\lmmseChannel^{(2)}_{jju}[n]]_{p}[\lmmseChannel^{(1)}_{jju}[n]]^{*}_p, \quad \forall p \in {1 \dots M} \label{rawS} \;.
	\end{align}
	A regularized estimate for $\elementwiseR_{jju}$ is given by
	\begin{align}
	\regElementwiseRestimate_{jju} = \biasFactorForR \elementwiseRestimate_{jju} + (1-\biasFactorForR)\diag(\biasRMatrix). \label{regS}
	\end{align}
	We define $\regElementwiseRActual_{jju}$ as the expected value of $\regElementwiseRestimate_{jju}$ given as $\regElementwiseRActual_{jju} \triangleq  \mathbb{E}\{\regElementwiseRestimate_{jju}\} = \biasFactorForR \elementwiseR_{jju} + {(1-\biasFactorForR)\diag(\biasRMatrix)}$ for future use.
	
	In summary, the BS needs to compute channel covariance matrices for each set of $\tau_s$ blocks in order to obtain the LMMSE-type/element-wise LMMSE-type channel estimates in each coherence block within the set. 
	The quality of the LMMSE-type/element-wise LMMSE-type channel estimate and hence, SE of the system depends on the quality of the channel and covariance estimates, which in turn depend on parameters $\pilotLengthForR$, and $\pilotLengthForQ$. Therefore, it is crucial to study the impact of these parameters on a user's SE using a closed-form expression of SE. In the following section, we derive the SE results for both UL and DL data under the described LMMSE-type and element-wise LMMSE-type channel estimation.
	\vspace{-3mm}
	
	\section{Main Results: UL and DL Spectral Efficiency}
	\label{sec:mainResults}
	\subsection{Uplink Spectral Efficiency}
	In this section, the average SE for the UL channel of a target user $(j, u)$ is derived when the channel estimates are used in a matched-filter combiner at the BS. For the matched-filter, the combining vector $\mathbf{v}_{ju}[n]$ can be written as  $\mathbf{v}_{ju}[n] = \lmmseChannel_{jju}[n] = \unRegWEst_{ju}\leastSqrChannel_{jju}[n]$, where 
	\begin{equation*}
	\unRegWEst_{ju} = \begin{cases}
	\regRestimate_{jju}\rawQestimate^{-1}_{ju}, \; &\text{LMMSE-type channel estimate} \\
	\regElementwiseRestimate_{jju}\elementwiseQestimate^{-1}_{ju},  \qquad &\text{element-wise LMMSE-type channel estimate}.
	\end{cases}
	\end{equation*}
	
	For the sake of mathematical tractability, we make the following assumptions
	\begin{itemize}
		\item $\regRestimate_{jju}$ ($\regElementwiseRestimate_{jju}$) and $\rawQestimate_{ju}$ ($\elementwiseQestimate_{ju}$) are each computed from a different nonoverlapping set of coherence blocks that does not include $n^{th}$ block. Consequently, the random variables $\regRestimate_{jju} / \regElementwiseRestimate_{jju}$, $\rawQestimate_{ju} / \elementwiseQestimate_{ju}$, and $\leastSqrChannel_{jju}[n]$ are mutually uncorrelated.
		\item For the LMMSE-type channel estimate, $\pilotLengthForQ$ is assumed greater than $M$, so that the distribution of $\rawQestimate^{-1}_{ju}$ is non-degenerate inverse Wishart.
	\end{itemize}
	The received combined signal is given by
	\begin{align}
	\mathbf{v}_{ju}^H\mathbf{y}_j &= \sqrt{\ulSnr}\mathbb{E}\{\mathbf{v}_{ju}^H\mathbf{h}_{jju}\} x_{ju} +\sqrt{\ulSnr}(\mathbf{v}_{ju}^H\mathbf{h}_{jju} - \mathbb{E}\{\mathbf{v}_{ju}^H\mathbf{h}_{jju}\}) x_{ju} + \sum_{k\ne u} \sqrt{\ulSnr}\mathbf{v}_{ju}^H\mathbf{h}_{jjk} x_{jk} \nonumber \\
	&+ \sum_{l \ne j}\sum_{k=1}^{K} \sqrt{\ulSnr}\mathbf{v}_{ju}^H\mathbf{h}_{jlk} x_{lk} + \mathbf{v}_{ju}^H\mathbf{n}_j \quad .
	\label{eqn:combinerOut}
	\end{align}
	In \eqref{eqn:combinerOut}, the first term corresponds to the signal component, the second term is a result of the uncertainty in the array gain, the third term corresponds to the non-coherent intra-cell interference, the fourth term corresponds to the the coherent interference from pilot contamination, and the last term corresponds to the additive noise component. Since the first term is uncorrelated with the subsequent terms, a lower bound on SE of the UL channel from user $(j,u)$ to BS $j$ can be obtained as \cite{bjornson2016imperfect}
	\begin{align*}
	SE_{ju}^{(\mathrm{ul})} = \left(1-\frac{P}{\ulCoherenceSymbols} - \frac{\pilotLengthForR P}{\ulCoherenceSymbols\tau_s}\right)\log_2\left(1+\gamma_{ju}^{(ul)}\right),&&[bits/s/Hz]
	\end{align*}
	where $\gamma_{ju}^{(\mathrm{ul})}$ is given by
	\begin{align*}
	\gamma_{ju}^{(\mathrm{ul})} &= \frac{|\mathbb{E}\{\mathbf{v}_{ju}^H\mathbf{h}_{jju}\}|^2}
	{\sum\limits_{l=1}^{L}\sum\limits_{k=1}^{K} \mathbb{E}\{|\mathbf{v}_{ju}^H\mathbf{h}_{jlk}|^2\} - |\mathbb{E}\{\mathbf{v}_{ju}^H\mathbf{h}_{jju}\}|^2 + \frac{1}{\ulSnr}\mathbb{E}\{\mathbf{v}_{ju}^H\mathbf{v}_{ju}\}}
	\end{align*}
	and the expectation $\mathbb{E}\{\cdot\}$ is over the channel realizations. In the pre-log factor, $P/\ulCoherenceSymbols$ accounts for ChEst pilots, and $\pilotLengthForR P/\ulCoherenceSymbols\tau_s$ accounts for CovEst pilots. However, since we assume that $\unRegWEst_{ju}$ and $\leastSqrChannel_{jju}[n]$ are mutually independent, we have $\mathbb{E}\{\cdot\} = \mathbb{E}_W\{\mathbb{E}_{h^{LS}}\{\cdot\}\}$, where $\mathbb{E}_W$ is the expectation over $\unRegWEst_{ju}$, and $\mathbb{E}_{h^{LS}}$ is the expectation over the LS estimate. The signal to noise ratio (SNR) expression can be further simplified to \cite{bjornson2016imperfect}
	\begin{align}
	&\gamma_{ju}^{(\mathrm{ul})} = \frac{|\mathbb{E}_W\{\trace(\genericWEst^H_{ju}\mathbf{R}_{jju})\}|^2}
	{\mathbb{E}_W\{\trace(\genericWEst_{ju}\mathbf{Q}_{ju}\genericWEst^H_{ju}\Rsum)\} +\sum\limits_{l=1}^{L}\mathbb{E}_W\{|\trace(\genericWEst^{H}_{ju}\mathbf{R}_{jlu})|^2\}-|\mathbb{E}_W\{\trace(\genericWEst^H_{ju}\mathbf{R}_{jju})\}|^2} \label{gamma}
	\end{align}
	where $ \Rsum \triangleq \sum\limits_{l=1}^{L}\sum\limits_{k=1}^{K}\mathbf{R}_{jlk} + \frac{1}{\ulSnr}\mathbf{I}$.
	
	\subsection{Uplink Spectral Efficiency when $\unRegWEst_{ju} = \regRestimate_{jju}\rawQestimate^{-1}_{ju}$}
	In this subsection, expressions for all the terms given in \eqref{gamma} are derived for the case when $\unRegWEst_{ju} = \regRestimate_{jju}\rawQestimate^{-1}_{ju}$. In what follows, $\mathbb{E}_R\{\cdot\}$ represents the expectation over  $\regRestimate_{jju}$, $\mathbb{E}_Q\{\cdot\}$ represents the  expectation over $\rawQestimate_{ju}$, and $\mathbb{E}_W\{\cdot\}$ represents the expectation over both $\regRestimate_{jju}$ and $\rawQestimate_{ju}$. It should be noted that, as already mentioned, we have assumed that $\regRestimate_{jju}$ and $\rawQestimate_{ju}$ are estimated from different pilot resources (coherence blocks) such that the estimates are independent to each other. Therefore, $\mathbb{E}_R\{\cdot\}$ and $\mathbb{E}_Q\{\cdot\}$ can be evaluated independently.
	
	Before analytically deriving the expectations for the terms in \eqref{gamma}, we present some useful lemmas.
	\vspace{-3mm}
	\begin{lemmas} \label{gaussianVectors}
		Given an arbitrary matrix $\arbSqrMatrix\in \mathbb{C}^{M\times M}$, and for any mutually independent M-dimensional random vector $\mathbf{h}$ distributed as $\mathcal{CN}(\mathbf{0},\mathbf{R})$, we have
		\begin{align}
		&\mathbb{E}\{\mathbf{h}\mathbf{h}^H\arbSqrMatrix\mathbf{h}\mathbf{h}^H\} = \mathbf{R}\arbSqrMatrix\mathbf{R} + \mathbf{R} \trace(\arbSqrMatrix\mathbf{R}) \label{hAh}\\
		&\mathbb{E}\{|\mathbf{h}^H\arbSqrMatrix\mathbf{h}|^2\} = |\trace(\arbSqrMatrix^H\mathbf{R})|^2 + \trace(\arbSqrMatrix\mathbf{R}\arbSqrMatrix^H\mathbf{R}).\label{abshAh}
		\end{align}
	\end{lemmas}
	\begin{proof}
		The proof is available in Appendix \ref{Appendix lemma 1}.
	\end{proof}
	\vspace{-3mm}
	\begin{lemmas} \label{Lemma:IWishproperties}
		Given a Hermitian matrix $\arbSymMatrix\in \mathbb{C}^{M\times M}$, an arbitrary matrix $\arbSqrMatrix\in \mathbb{C}^{M\times M}$,
		and a complex Wishart matrix, $\arbWishart \in \mathbb{C}^{M \times M}$, distributed as $\mathcal{W}(N,\mathbf{I})$, we have
		\begin{align}
		&\mathbb{E}\big\{[\arbWishart^{-1}]_{ij}\big\} = \frac{[\mathbf{I}]_{ij}}{N-M} \label{Wishart1}\\
		&\mathbb{E}\big\{[\arbWishart^{-1}]_{ij}[\arbWishart^{-1}]_{lk}\big\} = \frac{[\mathbf{I}]_{ij}[\mathbf{I}]_{lk} + \frac{1}{N-M}[\mathbf{I}]_{lj}[\mathbf{I}]_{ik}}{(N-M)^2-1} \label{Wishart2}\\
		&\mathbb{E}\{\trace(\arbWishart^{-2}\arbSymMatrix)\} = \frac{N}{(N-M)^3-(N-M)}\trace(\arbSymMatrix) \label{Wishart3}\\
		&\mathbb{E}\{|\trace(\arbWishart^{-1}\arbSqrMatrix)|^2\}
		=\frac{|\trace(\arbSqrMatrix)|^2 + \frac{1}{N-M}\trace(\arbSqrMatrix\arbSqrMatrix^{H})}{(N-M)^2-1}. \label{Wishart4}
		\end{align}
	\end{lemmas}
	\begin{proof}
		The proof is available in Appendix  \ref{Appendix lemma 2}.
	\end{proof}
	\vspace{-3mm}
	\begin{lemmas} \label{den1R}
		Given an arbitrary matrix $\arbSqrMatrix\in \mathbb{C}^{M\times M}$, we have
		\begin{align}
		&\mathbb{E}\{\rawRestimate_{jju}\arbSqrMatrix\rawRestimate_{jju}\} = \mathbf{R}_{jju}\arbSqrMatrix\mathbf{R}_{jju}
		+ \frac{1}{2\pilotLengthForR} \mathbf{Q}_{ju}\trace(\arbSqrMatrix\mathbf{Q}_{ju}) 
		+ \frac{1}{2\pilotLengthForR} \mathbf{R}_{jju}\trace(\arbSqrMatrix\mathbf{R}_{jju}) \label{RAR} \\
		&\mathbb{E}\{|\trace(\rawRestimate_{jju}\arbSqrMatrix)|^2\} = |\trace(\mathbf{R}_{jju}\arbSqrMatrix)|^2
		+ \frac{1}{2\pilotLengthForR}\trace(\arbSqrMatrix\mathbf{Q}_{ju}\arbSqrMatrix^H\mathbf{Q}_{ju}) 
		+ \frac{1}{2\pilotLengthForR}\trace(\arbSqrMatrix\mathbf{R}_{jju}\arbSqrMatrix^H\mathbf{R}_{jju}).\label{absTraceRA}
		\end{align}
	\end{lemmas}
	\begin{proof}
		The proof of this lemma uses Lemma~\ref{gaussianVectors} and is presented in Appendix \ref{Appendix lemma 3}.
	\end{proof}
	Now we are ready to formulate the key theorem of this subsection.
	\vspace{-3mm}
	\begin{theorems} \label{Theorem:raw_R}
		The numerator term of \eqref{gamma} when $\unRegWEst_{ju} = \regRestimate_{jju}\rawQestimate^{-1}_{ju}$ is given by
		\begin{footnotesize}
		\begin{align}
		&\mathbb{E}_W\{\trace(\unRegWEst^H_{ju}\mathbf{R}_{jju})\} = \trace(\mathbf{W}^H_{ju}\mathbf{R}_{jju}) + \left\{ \frac{\pilotLengthForQ}{\pilotLengthForQ-M}\trace(\bar{\mathbf{W}}^H_{ju}\mathbf{R}_{jju})-
		\trace(\mathbf{W}^H_{ju}\mathbf{R}_{jju})\right\} \label{num}
		\end{align}
		\end{footnotesize}
		The first and second terms of the denominator in \eqref{gamma} are given by
		\begin{footnotesize}
		\begin{align}
		&\mathbb{E}_W\{\trace(\unRegWEst_{ju}\mathbf{Q}_{ju}\unRegWEst^H_{ju}\Rsum)\}= \trace(\mathbf{W}_{ju}\mathbf{Q}_{ju}\mathbf{W}^H_{ju}\Rsum)\!
		+ \!\!\bigg\{\!\nCubeConstant \trace(\bar{\mathbf{W}}_{ju}\mathbf{Q}_{ju}\bar{\mathbf{W}}^H_{ju}\Rsum)\!
		-\!\trace(\mathbf{W}_{ju}\mathbf{Q}_{ju}\mathbf{W}^H_{ju}\Rsum)\!
		+ \!\frac{\biasFactorForR^2 \nCubeConstant}{2\pilotLengthForR} M \trace(\Rsum\mathbf{Q}_{ju}) \nonumber \\
		&+ \frac{\biasFactorForR^2 \nCubeConstant}{2\pilotLengthForR}\trace(\mathbf{W}_{ju})\trace(\Rsum \mathbf{R}_{jju})\bigg\} \label{den1}\\
		&\mathbb{E}_W\{|\trace(\unRegWEst^{H}_{ju}\mathbf{R}_{jlu})|^2\}\!
		= \!|\trace(\mathbf{W}^{H}_{ju}\mathbf{R}_{jlu})|^2\!
		+ \!\bigg\{\nSqareConstant|\trace(\bar{\mathbf{W}}^{H}_{ju}\mathbf{R}_{jlu})|^2\!
		-\!|\trace(\mathbf{W}^{H}_{ju}\mathbf{R}_{jlu})|^2\!
		+ \!\frac{\biasFactorForR^2 \nSqareConstant}{2\pilotLengthForR}\trace(\mathbf{W}_{lu}\mathbf{Q}_{ju}\mathbf{W}^{H}_{lu}\mathbf{Q}_{ju}) \nonumber \\
		&+ \!\frac{\biasFactorForR^2 \nSqareConstant}{2\pilotLengthForR}\trace(\mathbf{W}_{lu}\mathbf{R}_{jju}\mathbf{W}^{H}_{lu}\mathbf{R}_{jju})\!
		+ \!\frac{\nCubeConstant}{\pilotLengthForQ} \trace(\bar{\mathbf{W}}_{ju}^{2}\mathbf{Q}_{ju}\mathbf{W}^{2}_{lu}\mathbf{Q}_{ju})\!
		+ \!\frac{\biasFactorForR^2 \nCubeConstant}{2\pilotLengthForQ\pilotLengthForR} M\trace(\mathbf{W}^2_{jlu}\mathbf{Q}^2_{ju}) \nonumber\\
		&+ \!\frac{\biasFactorForR^2 \nCubeConstant}{2 \pilotLengthForQ\pilotLengthForR} \trace(\mathbf{W}_{ju})\trace(\mathbf{W}^2_{jlu}\mathbf{Q}_{ju}\mathbf{R}_{jju})\bigg\} \label{den2}
		\end{align}
		\end{footnotesize}
		where
		$\nCubeConstant \triangleq {\pilotLengthForQ\nSqareConstant/(\pilotLengthForQ-M)}$, $\nSqareConstant \triangleq \pilotLengthForQ^2/((\pilotLengthForQ-M)^2-1)$,
		$\unRegWActual_{ju} \triangleq \regRactual_{jju}\mathbf{Q}_{ju}^{-1}$ and ${\mathbf{W}_{lu} \triangleq \mathbf{R}_{jlu}\mathbf{Q}_{ju}^{-1}}$ for all $l = 1$ to $L$.
	\end{theorems}
	\begin{proof}
		We define a matrix $\wishartQ_{ju}$ as
		\begin{align}\label{WishartQ}
		&\wishartQ_{ju} \triangleq \pilotLengthForQ(\mathbf{Q}^{-\frac{1}{2}}_{ju}\rawQestimate_{ju}\mathbf{Q}^{-\frac{1}{2}}_{ju}).
		\end{align}
		It can be seen that $\wishartQ_{ju}$ is Wishart distributed, i.e., $\mathcal{W}(\pilotLengthForQ, \mathbf{I})$.
		
		Using \eqref{WishartQ} and the fact that $\unRegWEst_{ju} = \regRestimate_{jju}\rawQestimate^{-1}_{ju}$, the numerator term of \eqref{gamma} can be written as
		\begin{align}
		&\mathbb{E}_W\{\trace(\unRegWEst^H_{ju}\mathbf{R}_{jju})\}
		= \pilotLengthForQ\mathbb{E}_W\{\trace(\mathbf{Q}^{-\frac{1}{2}}_{ju}\wishartQ^{-1}_{ju}\mathbf{Q}^{-\frac{1}{2}}_{ju}\regRestimate_{jju}\mathbf{R}_{jju})\}. \label{num_interm}
		\end{align}
		Taking direct expectation over $\regRestimate_{jju}$ in \eqref{num_interm} and also using Lemma~\ref{Lemma:IWishproperties}, \eqref{num} can be obtained.
		
		Proof of \eqref{den1} and \eqref{den2} is as follows. Substituting $\unRegWEst_{ju} = \regRestimate_{jju}\rawQestimate^{-1}_{ju}$ into the first and second terms in the denominator of \eqref{gamma} and using Lemma~\ref{Lemma:IWishproperties}, we get the following equations
		\begin{align}
		&\mathbb{E}_W\{\trace(\unRegWEst_{ju}\mathbf{Q}_{ju}\unRegWEst^H_{ju}\Rsum)\} 
		= \nCubeConstant \mathbb{E}_R\{\trace(\mathbf{Q}^{-1}_{ju}\regRestimate_{jju}\Rsum\regRestimate_{jju})\} \label{den1_interm}\\
		&\mathbb{E}_W\!\{\!|\trace(\unRegWEst^{H}_{ju}\mathbf{R}_{jlu})|^2\!\}\! =\!\nSqareConstant\mathbb{E}_R\!\{\!|\trace(\mathbf{Q}^{-1}_{ju}\regRestimate_{jju}\mathbf{R}_{jlu})|^2\!\}\!
		+ \!\frac{\nCubeConstant}{\pilotLengthForQ}\!\mathbb{E}_R\!\{\!\trace(\mathbf{Q}^{-1}_{ju}\regRestimate_{jju}\mathbf{R}^{2}_{jlu}\regRestimate_{jju}\mathbf{Q}^{-1}_{ju})\!\}. \label{den2_interm}
		\end{align}
		Then using Lemma~\ref{den1R}, and substituting \eqref{regR} into \eqref{den1_interm} and \eqref{den2_interm}, we get \eqref{den1} and \eqref{den2}, respectively.
	\end{proof}
	Note that the expectation terms given in Theorem~\ref{Theorem:raw_R} contain two components: (i) the component that corresponds to known covariance information (first term of the right-hand side of the equations) and (ii) a penalty component (all terms except the first term of the right-hand side of the equations) due to regularization of $\mathbf{R}_{jju}$ and due to imperfect channel covariance information. For $\biasFactorForR = 1$, and as $\pilotLengthForR$ and $\pilotLengthForQ$ tend to infinity, the penalty components of the expectation terms vanish.
	\subsection{Uplink Spectral Efficiency when $\elementwiseWest_{ju} = \regElementwiseRestimate_{jju}\elementwiseQestimate^{-1}_{ju}$}
	In this subsection, derivations are presented for all the terms given in \eqref{gamma} when $\elementwiseWest_{ju} = \regElementwiseRestimate_{jju}\elementwiseQestimate^{-1}_{ju}$. In what follows, $\mathbb{E}_S\{\cdot\}$ represents the expectation over $\regElementwiseRestimate_{jju}$, $\mathbb{E}_P\{\cdot\}$ represents the expectation over $\elementwiseQestimate_{ju}$, and $\mathbb{E}_W\{\cdot\}$ represents the expectation over both $\regElementwiseRestimate_{jju}$ and $\elementwiseQestimate_{ju}$.
	
	Before analytically deriving the expectations for the terms in \eqref{gamma}, we present some useful lemmas. 
	\begin{lemmas} \label{h1h2}
		Given a zero mean complex Gaussian $2 \times 1$ random vector $\mathbf{h} =[h_{1}, h_{2}]^\intercal$ with covariance matrix
		$$\mathbf{R} =
		\begin{bmatrix}
		r_{11} & r_{12}\\
		r_{21} & r_{22}
		\end{bmatrix}$$
		we can state that $\mathbb{E}\{|h_1|^2|h_2|^2\} = r_{11} r_{22} + r_{12} r_{21}$.	
	\end{lemmas}
	\begin{proof}
		The proof of this lemma is straight forward to obtain and we omit it due to lack of space.
	\end{proof}
	\vspace{-3mm}
	\begin{lemmas} \label{el_IWishproperties}
		Given arbitrary matrices $\arbSqrMatrix_1 \in \mathbb{C}^{M\times M}$, ${\arbSqrMatrix_2 \in \mathbb{C}^{M\times M}}$, $\arbSqrMatrix \in \mathbb{C}^{M\times M}$, 
		and a matrix $\arbDiagonalWishart = \mathbf{Z}/2$, where $\mathbf{Z}$ is a diagonal matrix whose elements are i.i.d. $\chi^2$ random variables with $2N$-degrees of freedom ($N > 2$), we have
		\begin{align}
		&\mathbb{E}\{\trace(\arbDiagonalWishart^{-1} \arbSqrMatrix_1 \arbDiagonalWishart^{-1} \arbSqrMatrix_2)\}
		= \genericConstantOne_1 \trace(\arbSqrMatrix_1\arbSqrMatrix_2) + \genericConstantOne_2 \trace(\arbSqrMatrix_{1d}\arbSqrMatrix_{2d}) \label{YA1YA2}\\
		&\mathbb{E}\{|\trace(\arbDiagonalWishart^{-1}\arbSqrMatrix)|^2\} = \genericConstantOne_1 |\trace(\arbSqrMatrix)|^2 + \genericConstantOne_2 \trace(\arbSqrMatrix^H_d\arbSqrMatrix_d) \label{abstrace_YA}
		\end{align}
		where $\genericConstantOne_1 \triangleq 1/(N-1)^2$, $\genericConstantOne_2 \triangleq \genericConstantOne_1/(N-2)$, $\arbSqrMatrix_{1d} \triangleq \diag(\arbSqrMatrix_1)$, $\arbSqrMatrix_{2d} \triangleq \diag(\arbSqrMatrix_2)$, and $\arbSqrMatrix_d \triangleq \diag(\arbSqrMatrix)$.
	\end{lemmas}
	\begin{proof}
		The proof is available in Appendix \ref{Appendix lemma 5}.
	\end{proof}
	\vspace{-3mm}
	\begin{lemmas} \label{el_den1R}
		Given an arbitrary matrix $\arbSqrMatrix\in \mathbb{C}^{M\times M}$ and an arbitrary diagonal matrix $\arbDiaMatrix\in \mathbb{R}^{M\times M}$, then \vspace{-3mm}
		\begin{align}
		&\mathbb{E}\{\elementwiseRestimate_{jju}\arbSqrMatrix\elementwiseRestimate_{jju}\} = \elementwiseR_{jju}\arbSqrMatrix\elementwiseR_{jju}
		+\frac{1}{2\pilotLengthForR} \arbSqrMatrix\circ\mathbf{Q}_{ju}\circ \mathbf{Q}_{ju}
		+\frac{1}{2\pilotLengthForR} \arbSqrMatrix\circ\mathbf{R}_{jju}\circ \mathbf{R}_{jju} \label{SAS}\\
		&\mathbb{E}\{|\trace(\elementwiseRestimate_{jju}\arbDiaMatrix)|^2\}\! = \!|\trace(\elementwiseR_{jju}\arbDiaMatrix)|^2\!
		+ \!\frac{1}{2\pilotLengthForR} \!\sum_{p=1}^{M} \!\sum_{q=1}^{M}\!\Big\{\![\arbDiaMatrix\!(\!\mathbf{Q}_{ju}\!\circ \!\mathbf{Q}_{ju}\!)\!\arbDiaMatrix]_{pq}\!
		+ \![\arbDiaMatrix\!(\!\mathbf{R}_{jju}\!\circ \!\mathbf{R}_{jju}\!)\!\arbDiaMatrix]_{pq}\!\Big\}. \label{el_absTraceRA}
		\end{align}
	\end{lemmas}
	\begin{proof}
		The proof is available in Appendix \ref{Appendix lemma 6}.
	\end{proof}
	Now we are ready to formulate the key theorem of this subsection.
	\vspace{-3mm}
	\begin{theorems} \label{Theorem:raw_P}
		The numerator term of \eqref{gamma} when $\elementwiseWest_{ju} = \regElementwiseRestimate_{jju}\elementwiseQestimate^{-1}_{ju}$ is given by
		\begin{footnotesize}
		\begin{align}
		&\mathbb{E}_W\{\trace(\elementwiseWest^H_{ju}\mathbf{R}_{jju})\}
		= \trace(\elementwiseWActual^H_{ju}\mathbf{R}_{jju}) + \bigg\{\frac{\pilotLengthForQ}{(\pilotLengthForQ-1)}\trace(\elementwiseWActual^H_{ju}\mathbf{R}_{jju})
		- \trace(\mathbf{W}^H_{ju}\mathbf{R}_{jju}) \bigg\} \label{el_num}
		\end{align}
		\end{footnotesize}
		The first and second terms of the denominator in \eqref{gamma} are given by 
		\begin{footnotesize}
		\begin{align}
		&\mathbb{E}_W\!\{\trace(\elementwiseWest_{ju}\mathbf{Q}_{ju}\elementwiseWest^H_{ju}\Rsum)\}\! 
		= \!\trace(\mathbf{W}_{ju}\mathbf{Q}_{ju}\mathbf{W}^H_{ju}\Rsum\!)\!
		+ \!\bigg\{\!\elementWiseNSqrConst \trace(\elementwiseWActual_{ju}\mathbf{Q}_{ju}\elementwiseWActual^H_{ju}\Rsum\!)\!
		- \!\trace(\mathbf{W}_{ju}\mathbf{Q}_{ju}\mathbf{W}^H_{ju}\Rsum\!) \nonumber\\
		&+ \frac{\biasFactorForR^2 \elementWiseNSqrConst}{2\pilotLengthForR}\trace\Big(\!\elementwiseQ^{-1}_{ju}\mathbf{Q}_{ju}\elementwiseQ^{-1}_{ju}\{\Rsum\circ\mathbf{Q}_{ju}\circ \!\mathbf{Q}_{ju}\}
		+\!\elementwiseQ^{-1}_{ju}\mathbf{Q}_{ju}\elementwiseQ^{-1}_{ju}\!\{\!\Rsum\!\circ\!\mathbf{R}_{jju}\!\circ \!\mathbf{R}_{jju}\!\}\!\!\Big)\!\!
		+ \!\elementWiseNCubeConst \trace\!(\!\elementwiseWActual_{ju}\elementwiseQ_{ju}\!\elementwiseWActual^H_{ju}\elementWiseRsum\!)\!
		+ \!\frac{\biasFactorForR^2 \elementWiseNCubeConst}{2\pilotLengthForR}\!\trace(\elementWiseRsum\elementwiseQ_{ju}\!) \nonumber\\
		&+ \!\frac{\biasFactorForR^2 \elementWiseNCubeConst}{2\pilotLengthForR}\trace(\mathbf{W}_{ju}\elementWiseRsum\elementwiseR_{jju}\!) \bigg\} \label{el_den1}\\
		&\mathbb{E}_W\{|\trace(\elementwiseWest^{H}_{ju}\mathbf{R}_{jlu})|^2\} 
		= \!|\trace(\mathbf{W}^{H}_{ju}\elementwiseR_{jlu})|^2\!
		+ \!\bigg\{ \!\elementWiseNSqrConst |\trace(\elementwiseWActual^{H}_{ju}\elementwiseR_{jlu})|^2 - |\trace(\mathbf{W}^{H}_{ju}\elementwiseR_{jlu})|^2\!
		+\!\frac{\biasFactorForR^2 \elementWiseNSqrConst}{2\pilotLengthForR} \sum_{p=1}^{M}\sum_{q=1}^{M} [\mathbf{W}_{lu}(\mathbf{Q}_{ju}\circ\mathbf{Q}_{ju})\mathbf{W}_{lu}]_{pq} \nonumber \\
		&+\!\frac{\biasFactorForR^2 \elementWiseNSqrConst}{2\pilotLengthForR} \sum_{p=1}^{M}\sum_{q=1}^{M}\![\mathbf{W}_{lu}(\mathbf{R}_{jju}\circ\mathbf{R}_{jju})\mathbf{W}_{lu}]_{pq}\!
		+ \!\elementWiseNCubeConst \trace(\elementwiseWActual^2_{ju}\elementwiseR^2_{jlu})\!
		+\!\frac{\biasFactorForR^2 \elementWiseNCubeConst}{2\pilotLengthForR} \trace(\mathbf{W}^2_{lu} \elementwiseQ^2_{ju})\!
		+\!\frac{\biasFactorForR^2 \elementWiseNCubeConst}{2\pilotLengthForR} \trace(\mathbf{W}^2_{lu} \elementwiseR^2_{jju}) \bigg\}\label{el_den2}
		\end{align}
		\end{footnotesize}
		where
		$\elementWiseNSqrConst = {\pilotLengthForQ^2/(\pilotLengthForQ-1)^2}$, $\elementWiseNCubeConst =  \elementWiseNSqrConst/(\pilotLengthForQ-2)$, $\elementWiseRsum \triangleq \diag(\Rsum)$,
		$\unRegWActual_{ju} \triangleq \regElementwiseRActual_{jju}\elementwiseQ_{ju}^{-1}$ and $\mathbf{W}_{lu} \triangleq \elementwiseR_{jlu}\elementwiseQ_{ju}^{-1}$ for all $l = 1$ to $L$.
	\end{theorems}
	\begin{proof}
		We define the diagonal matrix $\elementWiseWishartQ_{ju}$ as follows
		\begin{align}
		&\elementWiseWishartQ_{ju} \triangleq \pilotLengthForQ(\elementwiseQ^{-1}_{ju}\elementwiseQestimate_{ju}). \label{el_WishartQ}
		\end{align}
		It can be seen that the elements of $2\elementWiseWishartQ_{ju}$ are i.i.d. $\chi^2$ random variables with $2N$-degrees of freedom.
		
		Using \eqref{el_WishartQ} and the fact that $\elementwiseWest_{ju} = \regElementwiseRestimate_{jju}\elementwiseQestimate^{-1}_{ju}$, the numerator term of \eqref{gamma} can be written as
		\begin{align}
		\mathbb{E}_W\{\trace(\elementwiseWest^H_{ju}\mathbf{R}_{jju})\} &= \pilotLengthForQ \mathbb{E}_W\{\trace(\elementWiseWishartQ^{-1}_{ju}\elementwiseQ^{-1}_{ju}\regElementwiseRestimate_{jju}\mathbf{R}_{jju})\}  \nonumber\\
		&= \pilotLengthForQ\sum_{p=1}^{M} \mathbb{E}_P\{[\elementWiseWishartQ^{-1}_{ju}]_{pp}\}\mathbb{E}_S\{[\elementwiseQ^{-1}_{ju}\regElementwiseRestimate_{jju}\mathbf{R}_{jju}]_{pp}\}. \label{el_num_interm}
		\end{align}
		Taking direct expectation over $\regElementwiseRestimate_{jju}$ in \eqref{el_num_interm} and using the properties of inverse $\chi^2$ distribution, \eqref{el_num} can be obtained.
		
		Proof of \eqref{el_den1} and \eqref{el_den2} is as follows. Substituting $\elementwiseWest_{ju} = \regElementwiseRestimate_{jju}\elementwiseQestimate^{-1}_{ju}$ and \eqref{el_WishartQ} into the first and second denominator terms of \eqref{gamma} and using Lemma~\ref{el_IWishproperties}, we get the following equations
		\begin{align}
		&\mathbb{E}_W\{\!\trace(\elementwiseWest_{ju}\mathbf{Q}_{ju}\elementwiseWest^H_{ju}\Rsum\!)\!\}\!
		= \!\elementWiseNSqrConst \mathbb{E}_S\{\!\trace(\elementwiseQ^{-1}_{ju}\mathbf{Q}_{ju}\elementwiseQ^{-1}_{ju}\regElementwiseRestimate_{jju}\Rsum\regElementwiseRestimate_{jju}\!)\!\}\!
		+ \!\elementWiseNCubeConst \mathbb{E}_S\{\!\trace(\elementwiseQ^{-1}_{ju}\regElementwiseRestimate_{jju}\elementWiseRsum\regElementwiseRestimate_{jju}\!)\!\} \label{el_den1_interm}\\
		&\mathbb{E}_W\{|\trace(\unRegWEst^{H}_{ju}\mathbf{R}_{jlu})|^2\} = \elementWiseNSqrConst \mathbb{E}_S\{|\trace(\elementwiseQ^{-1}_{ju}\regElementwiseRestimate_{jju}\elementwiseR_{jlu})|^2\}
		+ \elementWiseNCubeConst \mathbb{E}_S\{\trace(\elementwiseQ^{-2}_{ju}\regElementwiseRestimate^2_{jju}\elementwiseR^2_{jlu})\}. \label{el_den2_interm}
		\end{align}
		
		Then using Lemma~\ref{el_den1R} and substituting \eqref{regS} into \eqref{el_den1_interm} and \eqref{el_den2_interm}, we get \eqref{el_den1} and \eqref{el_den2}, respectively.
	\end{proof}
	Similar to Theorem~\ref{Theorem:raw_R}, the penalty components of the expectation terms given in Theorem~\ref{Theorem:raw_P} also vanish if $\biasFactorForR = 1$, and as $\pilotLengthForR$ and $\pilotLengthForQ$ tend to infinity,.
	\vspace{-6mm}
	\subsection{Uplink Spectral Efficiency when $\elementwiseWest_{ju} = \regElementwiseRestimate_{jju}\elementwiseQestimate^{-1}_{ju}$ with Regularized $\elementwiseQestimate_{ju}$}
	In this section, derivations for all the terms given in \eqref{gamma} for element-wise LMMSE-type channel estimation with regularized $\elementwiseQestimate_{ju}$, are presented. The regularized estimate of $\elementwiseQ_{ju}$ is given by \vspace{-3mm}
	\begin{align}
	\regElementwiseQestimate_{ju} = \biasFactorForQ \unRegElementwiseQestimate_{ju} + (1-\biasFactorForQ)\biasQMatrix \label{regQ}
	\end{align}
	where $[\unRegElementwiseQestimate_{ju}]_{pp} = \frac{1}{\pilotLengthForQ}\sum_{n=1}^{\pilotLengthForQ}|[\lmmseChannel^{LS}_{jju}[n]]_p|^2, \, {\forall p \in {1 \dots M}}$ is the unbiased estimate of $\elementwiseQ_{ju}$; $\biasQMatrix$ is an arbitrary diagonal bias-matrix with positive elements; and $\biasFactorForQ$ is a design parameter. Furthermore, let us define the matrix $\elementWiseWishartQ_{ju} \triangleq \pilotLengthForQ(\elementwiseQ^{-1}_{ju}\unRegElementwiseQestimate_{ju})$ such that the elements of $2 \elementWiseWishartQ_{ju}$ are $\chi^2$ distributed with $2\pilotLengthForQ$ degrees of freedom. Now, we define two diagonal matrices, $\expectationMatrixFirstOrder$ and $\expectationMatrixSecondOrder$, whose elements are given by \vspace{-3mm}
	\begin{footnotesize}
	\begin{align}
	[\expectationMatrixFirstOrder]_{pp} &\triangleq \mathbb{E}\{[\regElementwiseQestimate_{ju}^{-1}]_{pp}\} =\mathbb{E}\left\{\left(\frac{1}{\pilotLengthForQ}\biasFactorForQ [\elementwiseQ_{ju}]_{pp}[\elementWiseWishartQ_{ju}]_{pp} + (1-\biasFactorForQ)[\biasQMatrix]_{pp} \right)^{-1}\right\} \label{firstOrderExpectation}\\
	[\expectationMatrixSecondOrder]_{pp} &\triangleq \mathbb{E}\{[\regElementwiseQestimate_{ju}^{-1}]^2_{pp}\}
	=\mathbb{E}\left\{\left(\frac{1}{\pilotLengthForQ}\biasFactorForQ [\elementwiseQ_{ju}]_{pp}[\elementWiseWishartQ_{ju}]_{pp} + (1-\biasFactorForQ)[\biasQMatrix]_{pp} \right)^{-2}\right\}. \label{secondOrderExpectation}
	\end{align}
	\end{footnotesize}
	It should be noted that expectation terms in the above equations can be evaluated numerically using the probability distribution function of the $\chi^2$ distribution. Therefore, SE expressions we derive in this section are not in a proper closed-form but involves matrices that can be computed numerically.
	
	Before deriving the expectations for the terms in \eqref{gamma}, we present a useful lemma.
	\vspace{-3mm}
	\begin{lemmas} \label{el_IWishpropertiesReg}
		Given arbitrary matrices $\arbSqrMatrix_1 \in \mathbb{C}^{M\times M}$, ${\arbSqrMatrix_2 \in \mathbb{C}^{M\times M}}$, $\arbSqrMatrix \in \mathbb{C}^{M\times M}$, 
		we have \vspace{-3mm}
		\begin{align}
		&\mathbb{E}\{\trace(\regElementwiseQestimate^{-1}_{ju} \arbSqrMatrix_1 \regElementwiseQestimate^{-1}_{ju} \arbSqrMatrix_2)\}
		=  \trace\left(\expectationMatrixFirstOrder\arbSqrMatrix_1\expectationMatrixFirstOrder\arbSqrMatrix_2\right) +  \trace\left((\expectationMatrixSecondOrder-\expectationMatrixFirstOrder^2)\arbSqrMatrix_{1d}\arbSqrMatrix_{2d}\right)  \label{YA1YA2Reg}\\
		&\mathbb{E}\{|\trace(\regElementwiseQestimate^{-1}\arbSqrMatrix)|^2\} = |\trace(\expectationMatrixFirstOrder \arbSqrMatrix)|^2 + \trace\left((\expectationMatrixSecondOrder-\expectationMatrixFirstOrder^2)\arbSqrMatrix^H_d \arbSqrMatrix_d\right) \label{abstrace_YAReg}
		\end{align}
		where $\arbSqrMatrix_{1d} \triangleq \diag(\arbSqrMatrix_1)$, $\arbSqrMatrix_{2d} \triangleq \diag(\arbSqrMatrix_2)$, and $\arbSqrMatrix_d \triangleq \diag(\arbSqrMatrix)$.
	\end{lemmas}
	\begin{proof}
		The proof is available in Appendix \ref{Appendix lemma 7}.
	\end{proof}
	Now we are ready to formulate the key theorem of this subsection.
	\vspace{-3mm}
	\begin{theorems} \label{Theorem:regularized_P}
		The numerator term of \eqref{gamma} when $\elementwiseWest_{ju} = \regElementwiseRestimate_{jju}\elementwiseQestimate^{-1}_{ju}$ is given by
		\begin{footnotesize}
		\begin{align}
		&\mathbb{E}_W\{\trace(\elementwiseWest^H_{ju}\mathbf{R}_{jju})\} = \trace(\expectationMatrixFirstOrder \elementwiseR_{jju}\mathbf{R}_{jju}). \label{el_numReg}
		\end{align}
		\end{footnotesize}
		The first and second terms of the denominator in \eqref{gamma} are given by
		\begin{footnotesize}
		\begin{align}
		&\mathbb{E}_W\{\trace(\elementwiseWest_{ju}\mathbf{Q}_{ju}\elementwiseWest^H_{ju}\Rsum)\}
		= \!\trace(\expectationMatrixFirstOrder\mathbf{Q}_{ju}\expectationMatrixFirstOrder\elementwiseR_{jju}\Rsum\elementwiseR_{jju})\!
		+ \!\frac{\biasFactorForR^2}{2\pilotLengthForR} \trace\Big(\!\expectationMatrixFirstOrder\mathbf{Q}_{ju}\expectationMatrixFirstOrder \{\Rsum\circ\mathbf{Q}_{ju}\circ \mathbf{Q}_{ju}\}\!
		+ \!\expectationMatrixFirstOrder\mathbf{Q}_{ju}\expectationMatrixFirstOrder \{\Rsum\circ\mathbf{R}_{jju}\circ \mathbf{R}_{jju}\}\!\Big)\! \nonumber\\
		&+ \left(1+\frac{\biasFactorForR^2}{2\pilotLengthForR}\right)\trace((\expectationMatrixSecondOrder-\expectationMatrixFirstOrder^2)\elementwiseQ_{ju} \elementwiseR^2_{jju} \elementWiseRsum)
		+ \frac{\biasFactorForR^2}{2\pilotLengthForR}\trace((\expectationMatrixSecondOrder-\expectationMatrixFirstOrder^2)\elementwiseQ^3_{ju} \elementWiseRsum)  \label{el_den1Reg}\\
		&\mathbb{E}_W\{|\trace(\unRegWEst^{H}_{ju}\mathbf{R}_{jlu})|^2\} = |\trace(\expectationMatrixFirstOrder\elementwiseR_{jju}\elementwiseR_{jlu})|^2\!
		+ \!\frac{\biasFactorForR^2}{2\pilotLengthForR} \sum_{p=1}^{M}\sum_{q=1}^{M} \Big\{[\expectationMatrixFirstOrder\elementwiseR_{jlu}(\mathbf{Q}_{ju}\circ\mathbf{Q}_{ju})\expectationMatrixFirstOrder\elementwiseR_{jlu}]_{pq}\!
		+\![\expectationMatrixFirstOrder\elementwiseR_{jlu}(\mathbf{R}_{jju}\circ\mathbf{R}_{jju})\expectationMatrixFirstOrder\elementwiseR_{jlu}]_{pq} \Big\}\! \nonumber \\
		&+ \left(1+\frac{\biasFactorForR^2}{2\pilotLengthForR}\right)\trace((\expectationMatrixSecondOrder-\expectationMatrixFirstOrder^2)\elementwiseR^2_{jju}\elementwiseR^2_{jlu})
		+ \frac{\biasFactorForR^2}{2\pilotLengthForR} \trace((\expectationMatrixSecondOrder-\expectationMatrixFirstOrder^2) \elementwiseQ^2_{ju}\elementwiseR^2_{jlu}).\label{el_den2Reg}
		\end{align}
		\end{footnotesize}
	\end{theorems}
	\begin{proof}	
		Using $\elementwiseWest_{ju} = \regElementwiseRestimate_{jju}\elementwiseQestimate^{-1}_{ju}$, the numerator term of \eqref{gamma} can be written as
		\begin{align}
		\mathbb{E}_W\{\trace(\elementwiseWest^H_{ju}\mathbf{R}_{jju})\} &= \mathbb{E}_W\{\trace(\regElementwiseQestimate^{-1}_{ju}\regElementwiseRestimate_{jju}\mathbf{R}_{jju})\}. \label{el_num_intermReg}
		\end{align}
		Taking direct expectation over $\regElementwiseRestimate_{jju}$ in \eqref{el_num_intermReg} and also using \eqref{firstOrderExpectation}, \eqref{el_numReg} can be obtained.
		
		Proof of \eqref{el_den1Reg} and \eqref{el_den2Reg} is as follows. Substituting $\elementwiseWest_{ju} = \regElementwiseRestimate_{jju}\elementwiseQestimate^{-1}_{ju}$ into the first and second denominator terms of \eqref{gamma} and using Lemma~\ref{el_IWishpropertiesReg}, we get the following equations
		\begin{align}
		&\mathbb{E}_W\!\{\!\trace(\elementwiseWest_{ju}\mathbf{Q}_{ju}\elementwiseWest^H_{ju}\Rsum)\!\}\!
		= \!\mathbb{E}_S\{\!\trace(\expectationMatrixFirstOrder\mathbf{Q}_{ju}\expectationMatrixFirstOrder\regElementwiseRestimate_{jju}\Rsum\regElementwiseRestimate_{jju})\!\}\!
		+ \!\mathbb{E}_S\{\!\trace\!(\!(\!\expectationMatrixSecondOrder-\expectationMatrixFirstOrder^2)\!\elementwiseQ_{ju}\regElementwiseRestimate_{jju}\elementWiseRsum\regElementwiseRestimate_{jju}\!)\!\} \label{el_den1_intermReg}\\
		&\mathbb{E}_W\{|\trace(\unRegWEst^{H}_{ju}\mathbf{R}_{jlu})|^2\} =  \mathbb{E}_S\{|\trace(\expectationMatrixFirstOrder\regElementwiseRestimate_{jju}\elementwiseR_{jlu})|^2\}\!
		+ \!\mathbb{E}_S\{\trace((\expectationMatrixSecondOrder\!-\!\expectationMatrixFirstOrder^2)\regElementwiseRestimate^2_{jju}\elementwiseR^2_{jlu})\}. \label{el_den2_intermReg}
		\end{align}
		
		Then using Lemma~\ref{el_den1R} and substituting \eqref{regS} into \eqref{el_den1_intermReg} and \eqref{el_den2_intermReg}, we get \eqref{el_den1Reg} and \eqref{el_den2Reg}, respectively.	
	\end{proof}
	\vspace{-6mm}
	\subsection{Downlink Spectral Efficiency}
	The DL spectral efficiency for user $(j,u)$ is given in this section for a matched filter precoder, i.e.,  $\mathbf{b}_{ju} = \lmmseChannel_{jju}[n]/\sqrt{\mathbb{E}\{\norm{\lmmseChannel_{jju}[n]}^2\}} = \genericWEst_{ju}\leastSqrChannel_{jju}/\sqrt{\mathbb{E}\{\norm{\genericWEst_{ju}\leastSqrChannel_{jju}[n]}^2\}}$. Therefore, the received signal at user $(j,u)$ can be written as
	\begin{align}
	z_{ju} &= \sqrt{\dlSnr}\mathbb{E}\{\mathbf{b}_{ju}^H\mathbf{h}_{jju}\} d_{ju} +\sqrt{\dlSnr}(\mathbf{b}_{ju}^H\mathbf{h}_{jju} - \mathbb{E}\{\mathbf{b}_{ju}^H\mathbf{h}_{jju}\}) d_{ju} + \sum_{k\ne u} \sqrt{\dlSnr}(\mathbf{b}_{ju}^H\mathbf{h}_{jjk}) d_{jk} \nonumber \\
	&+ \sum_{l \ne j}\sum_{k=1}^{K} \sqrt{\dlSnr}(\mathbf{b}_{ju}^H\mathbf{h}_{jlk}) d_{lk} + e_{ju}.
	\label{eqn:dlReceived}
	\end{align}
	Here, we assume that the scalar in the denominator of the precoding vector, $\sqrt{\mathbb{E}\{\norm{\lmmseChannel_{jju}[n]}^2\}}$, is a known constant at the BS. The first term in \eqref{eqn:dlReceived} corresponds to the desired signal component, the second term corresponds to the uncertainty in the DL transmit array gain, the third term corresponds to the non-coherent intra-cell interference, the coherent interference from pilot contamination given by the fourth term, and the last term represents the additive noise component. The second term in \eqref{eqn:dlReceived} corresponds to the uncertainty in the DL transmit array gain. Then, due to the similarity between the UL received combined signal in \eqref{eqn:combinerOut} to the DL received signal, a lower bound on DL channel SE of the user $(j,u)$ can be easily shown to be
	\begin{align*}
	\mathrm{SE}_{ju}^{(dl)} = \log_2\left(1+\gamma_{ju}^{(dl)}\right)&&[bits/s/Hz],
	\end{align*}
	where $\gamma^{(dl)}_{ju}$ is given by
	\begin{align}
	&\gamma_{ju}^{(dl)}\! = \!\frac{|\mathbb{E}_W\{\trace(\genericWEst^H_{ju}\mathbf{R}_{jju})\}|^2}
	{\mathbb{E}_W\{\trace(\genericWEst_{ju}\mathbf{Q}_{ju}\genericWEst^H_{ju}\dlRsum)\}\! +\!\sum\limits_{l=1}^{L}\!\mathbb{E}_W\{|\trace(\genericWEst^{H}_{ju}\mathbf{R}_{jlu})|^2\}\!-\!|\mathbb{E}_W\{\trace(\genericWEst^H_{ju}\mathbf{R}_{jju})\}|^2\! + \!\frac{1}{\dlSnr}} \label{dlGamma}
	\end{align}
	and $ \dlRsum \triangleq \sum\limits_{l=1}^{L}\sum\limits_{k=1}^{K}\mathbf{R}_{jlk}$. We utilize channel hardening and avoid the use of DL pilots. Consequently, there is no pre-log factor for the SE expression. The expectation taken in all the terms of \eqref{dlGamma} is over the random matrix $\genericWEst_{ju}$. However, $\genericWEst_{ju} = \regRestimate_{jju}\rawQestimate_{ju}^{-1}$ for the LMMSE-type channel estimation and $\genericWEst_{ju} = \regElementwiseRestimate_{jju}\elementwiseQestimate_{ju}^{-1}$ for the element-wise LMMSE-type channel estimation. These expectation terms are already presented in Theorems \ref{Theorem:raw_R}, \ref{Theorem:raw_P}, and \ref{Theorem:regularized_P} for the LMMSE-type, element-wise LMMSE-type, and element-wise LMMSE-type with regularized  $\elementwiseQestimate_{ju}$ channel estimates, respectively. Furthermore, $\Rsum$ should be replaced by $\dlRsum$ in computing the expectation terms.
	
	\section{Discussion}
	\label{sec:discussion}
	The question of practical significane is the following. Based on the above obtained closed-form relations between the average SE and the parameters $\pilotLengthForR$ and $\pilotLengthForQ$, how to choose these parameters to provide a required SE? Thus, we discuss here the impact of these parameters on SE corresponding to the LMMSE-type and element-wise LMMSE-type channel estimation. We also compare the theoretical SE expressions for the LMMSE-type and element-wise LMMSE-type channel estimations. Since the SE expression for the element-wise LMMSE-type channel estimation with regularized $\elementwiseQestimate_{ju}$ is not in closed-form, we omit the analytic discussion for this case here and study it numerically in the next section.
	
	It can be noted from the expectation terms in Theorems~\ref{Theorem:raw_R} and \ref{Theorem:raw_P} that the penalty components due to imperfect covariance information gradually vanish with an increase in $\pilotLengthForR$ and $\pilotLengthForQ$, but the penalty due to the regularization remains finite. Furthermore, if $||\mathbf{W}_{ju} - \regWActual_{ju}||/||\mathbf{W}_{ju}||<<1$ (i.e., if $\biasFactorForR$ is close to $1$), one can state that these expectation terms converge to the values that correspond to the known covariance case. However, despite leading to an improvement in $\gamma_{ju}^{(ul)}$ (due to convergence of the expected values), an increase in $\pilotLengthForR$ also causes a degradation in the pre-log factor of the derived UL SE expression. Therefore, the choice of $\pilotLengthForR$ impacts UL SE in two ways: (i) smaller the value of $\pilotLengthForR$, higher the error in covariance estimation and hence lower the value of UL SE and (ii) larger the value of $\pilotLengthForR$, higher the consumption of UL resources and hence lower the value of UL SE. Whereas, due to the absence of DL pilots, the DL SE does not degrade with an increase in $\pilotLengthForR$; it gradually rises to the DL SE value that corresponds to the known covariance case. Larger $\pilotLengthForQ$ makes both the UL and DL SE better due to the improved estimates of $\mathbf{Q}_{ju}$ (or $\elementwiseQ_{ju}$). Therefore, given an SE requirement, the aim here is to choose minimum $\pilotLengthForR$ and $\pilotLengthForQ$ values that are sufficient to provide the desired SE.
	
	Since estimating $\mathbf{Q}_{ju}$ (or $\elementwiseQ_{ju}$) does not involve additional pilot transmission, choosing $\pilotLengthForQ$ is not as critical as choosing $\pilotLengthForR$. Therefore, if we consider $\pilotLengthForQ$ as known, it is also important to derive $\pilotLengthForR$ values that make the LMMSE-type channel estimation preferable to the element-wise LMMSE-type one, and vice-versa. By comparing the UL/DL SINR values (in \eqref{gamma} or \eqref{dlGamma}) for the two channel estimation techniques, we can compute a threshold value for $\pilotLengthForR$ ($\pilotLengthThreshold$), such that the element-wise LMMSE-type estimator is preferable if $\pilotLengthForR < \pilotLengthThreshold$, and the LMMSE-type estimator is preferable otherwise. Note that $\pilotLengthThreshold$ is different for UL and DL covariance estimation. It can be obtained by equating the SINR expressions for the LMMSE-type and element-wise LMMSE-type channel estimation techniques (for UL and DL) and solving the corresponding equation for $\pilotLengthForR$. After some straight forward algebra, $\pilotLengthThreshold$ can be obtained in the form:
	\vspace{-2mm}
	\begin{align}
		\pilotLengthThreshold = \frac{fc-ah}{ag-fb} \label{nThreshold}
	\end{align}
	where
	\begin{footnotesize}
	\begin{align*}
		a &= \left(\frac{\pilotLengthForQ}{\pilotLengthForQ-M}\trace(\bar{\mathbf{W}}^H_{ju}\mathbf{R}_{jju})\right)^2;
		\quad f = \left(\frac{\pilotLengthForQ}{(\pilotLengthForQ-1)}\trace(\elementwiseWActual^H_{ju}\mathbf{R}_{jju}) \right)^2\\
		b &= \nCubeConstant \trace(\bar{\mathbf{W}}_{ju}\mathbf{Q}_{ju}\bar{\mathbf{W}}^H_{ju}\Rsumgen)
		+ \sum_{l=1}^{L} \left\{\nSqareConstant|\trace(\bar{\mathbf{W}}^{H}_{ju}\mathbf{R}_{jlu})|^2
		+ \frac{\nCubeConstant}{\pilotLengthForQ} \trace(\bar{\mathbf{W}}_{ju}^{2}\mathbf{Q}_{ju}\mathbf{W}^{2}_{lu}\mathbf{Q}_{ju})\right\}
		- a + d\\
		c &= \frac{\biasFactorForR^2 \nCubeConstant}{2} \left\{M \trace(\Rsumgen\mathbf{Q}_{ju}) + \trace(\mathbf{W}_{ju})\trace(\Rsumgen \mathbf{R}_{jju})\right\}
		+ \frac{\biasFactorForR^2 \nSqareConstant}{2}\sum_{l=1}^{L}\big\{\trace(\mathbf{W}_{lu}\mathbf{Q}_{ju}\mathbf{W}^{H}_{lu}\mathbf{Q}_{ju})
		+ \trace(\mathbf{W}_{lu}\mathbf{R}_{jju}\mathbf{W}^{H}_{lu}\mathbf{R}_{jju})\big\}\\
		&+\frac{\biasFactorForR^2 \nCubeConstant}{2\pilotLengthForQ} \sum_{l=1}^{L}\left\{M\trace(\mathbf{W}^2_{jlu}\mathbf{Q}^2_{ju})
		+ \trace(\mathbf{W}_{ju})\trace(\mathbf{W}^2_{jlu}\mathbf{Q}_{ju}\mathbf{R}_{jju})\right\}\\
		g &= \!\elementWiseNSqrConst \!\!\left\{\!\!\trace\!(\!\elementwiseWActual_{ju}\mathbf{Q}_{ju}\elementwiseWActual^H_{ju}\Rsumgen\!)\! + \!\sum_{l=1}^{L}\!|\trace\!(\elementwiseWActual^{H}_{ju}\elementwiseR_{jlu})\!|^2\!\!\right\}\!\!
		+ \!\elementWiseNCubeConst \!\!\left\{\!\!\trace\!(\elementwiseWActual_{ju}\elementwiseQ_{ju}\elementwiseWActual^H_{ju}\Ssumgen)\! + \!\sum_{l=1}^{L}\!\trace\!(\elementwiseWActual^2_{ju}\elementwiseR^2_{jlu})\!\!\right\}\!\!
		-\!f\! + \!d\\
		h &= \frac{\biasFactorForR^2 \elementWiseNSqrConst}{2}\trace\Big(\!\elementwiseQ^{-1}_{ju}\mathbf{Q}_{ju}\elementwiseQ^{-1}_{ju}\{\Rsumgen\circ\mathbf{Q}_{ju}\circ \!\mathbf{Q}_{ju}\}
		+\!\elementwiseQ^{-1}_{ju}\mathbf{Q}_{ju}\elementwiseQ^{-1}_{ju}\!\{\!\Rsumgen\!\circ\!\mathbf{R}_{jju}\!\circ \!\mathbf{R}_{jju}\!\}\!\!\Big)\!\!
		+ \!\frac{\biasFactorForR^2 \elementWiseNCubeConst}{2}\Big\{\!\trace(\Ssumgen\elementwiseQ_{ju}\!)
		+ \!\trace(\mathbf{W}_{ju}\Ssumgen\elementwiseR_{jju}\!)\! \\
		&+ \!\trace(\mathbf{W}^2_{lu} \elementwiseQ^2_{ju}) + \trace(\mathbf{W}^2_{lu} \elementwiseR^2_{jju})\Big\}
		+ \frac{\biasFactorForR^2 \elementWiseNSqrConst}{2} \sum_{p=1}^{M}\sum_{q=1}^{M} \Big\{[\mathbf{W}_{lu}(\mathbf{Q}_{ju}\circ\mathbf{Q}_{ju})\mathbf{W}_{lu}]_{pq}
		+ \![\mathbf{W}_{lu}(\mathbf{R}_{jju}\circ\mathbf{R}_{jju})\mathbf{W}_{lu}]_{pq}\Big\}\!\\
		\Ssumgen &= \diag(\Rsumgen);
		\quad \quad \Rsumgen = \begin{cases}
		\Rsum, \; &\text{for UL} \\
		\dlRsum,  \qquad &\text{for DL}
		\end{cases};
		\quad \quad d = \begin{cases}
		0, \; &\text{for UL} \\
		\frac{1}{\dlSnr},  \qquad &\text{for DL.}
		\end{cases}
	\end{align*}
	\end{footnotesize}
	Note that $\pilotLengthThreshold$ is a function of $\pilotLengthForQ$ which can take any real value. Thus, if $\pilotLengthThreshold$ is negative for some value of $\pilotLengthForQ$, it means, for that particular choice of $\pilotLengthForQ$, there is no valid $\pilotLengthForR$ that makes the LMMSE-type channel estimation preferable. Consequently, using \eqref{nThreshold}, we can also compute a threshold for $\pilotLengthForQ$ below which element-wise LMMSE-type channel estimation is always preferred. However, deriving a theoretical expression for such a threshold is extremely difficult. It can be easily computed numerically.
	
	Therefore, the closed-form expressions for average UL and DL SE, for the LMMSE-type and element-wise LMMSE-type channel estimation methods serve as tools for choosing different design parameters, and also as a tool for choosing a preferred channel estimation technique. In practice, with approximate models of the covariance matrix of an individual user in a massive MIMO system, the derived expressions for average SE enables us to choose these parameters for the desired UL and DL SE values.
	
	In what follows, we provide a comparison of the derived theoretical SE expressions with simulated SE obtained by averaging over multiple realizations of random covariance estimation matrices. We also compare the theoretical SE expressions with the SE expressions that correspond to known covariance case. Finally, we also depict the impact of $\pilotLengthForR$ on the SE. 
	
	\section{Simulation Results}
	We consider a massive MIMO system with $L = 7$ cells, each comprising a BS with $M = 100$ antennas and $K = 10$ users. The BSs are at a distance of $300m$ apart from each other, and the users are uniformly spaced at a distance of $120m$ from the BS in their cells. The users that reuse the same pilot in different cells are at the same position relative to the corresponding BSs. Angular spread of the channel cluster is assumed to be $20^\circ$ within which the received paths from a user are assumed to be uniformly distributed. We consider the path loss model in \cite{7414036}, where the mean path loss is given as $ PL(f,d) = 20 \log_{10}\left( 4\pi f/c\right) + 10 n \log_{10}(d)$,	where $n$ is the path loss exponent, and $f$ is the operating frequency, and $c$ is the speed of light in $m/s$. Therefore, the mean received SNR, in dB, is given by $SNR = P_T - PL - 10\log_{10}(kT_o B) - NF$, where $P_T$ is the transmit power, $k$ is the Boltzmann constant, $T_0 = 290K$ is the nominal temperature, $B$ is the signal bandwidth, and $ NF $ is the noise figure in dB. In this setup, we consider $n= 3.76$, $f = 3.4$ GHz, $P_T = -3$ dBm, $B = 40 $ MHz, and $NF = 10$ dB, which results in the mean SNR of the received signal from a user that is at a distance $d$ from the BS to be given by $71.89-37.6 \log_{10} d$.
	
	The number of symbols that are dedicated for UL transmission within each coherence block is chosen to be $\ulCoherenceSymbols = 100$ symbols. We choose the number of symbols used for ChEst (and also CovEst) pilot to be $P=10$. Second-order statistics of the channel are assumed to be constant for $\tau_s = 25000$ coherence blocks, and the UL transmit power is $\ulSnr = 1$ and the DL transmit power is $\dlSnr = 10$. Additionally, we choose $\biasFactorForR = 0.95$, and $\biasRMatrix = \mathbf{I}$. For generating the regularized element-wise LMMSE-type channel estimation based SE, we choose $\biasFactorForQ = 0.95$. Sample averaging for all the expectation terms is computed using $500$ trials for different values of $\pilotLengthForR = \{125, 250, 500, 1000, 2000, 4000, 8000\}$.
	\subsection{Uplink Spectral Efficiency}
	\label{subsec: uplink simulations}
	For this simulation example, we consider the UL SE expressions that correspond to the three channel estimation techniques: LMMSE-type channel estimation, the element-wise LMMSE-type channel estimation, and the element-wise LMMSE-type channel estimation with regularized $\elementwiseQestimate_{ju}$. In Fig.~\ref{fig: UL SE}, we plot the SE as a function of $\pilotLengthForR$ for the three aforementioned channel estimation techniques. Fig.~\ref{fig: UL SE}(\subref{fig: UL SE 125}) depicts the SE values for $\pilotLengthForQ = 125$ and Fig.~\ref{fig: UL SE}(\subref{fig: UL SE 4000}) shows SE values for $\pilotLengthForQ = 4000$. In both the subplots, we present SE values corresponding to known covariance matrices (with no additional pilot overhead) and theoretical SE values as well as simulated SE values corresponding to the three channel estimation techniques that use the estimated covariance matrices.
	Note that the theoretical SE values for element-wise LMMSE-type channel estimation with regularized $\elementwiseQestimate_{ju}$ are computed numerically.
	
	\begin{figure}
		\begin{subfigure}{.5\textwidth}
			\centering
			\includegraphics[width=0.8\textwidth,trim={3cm 7cm 3cm 8cm}]{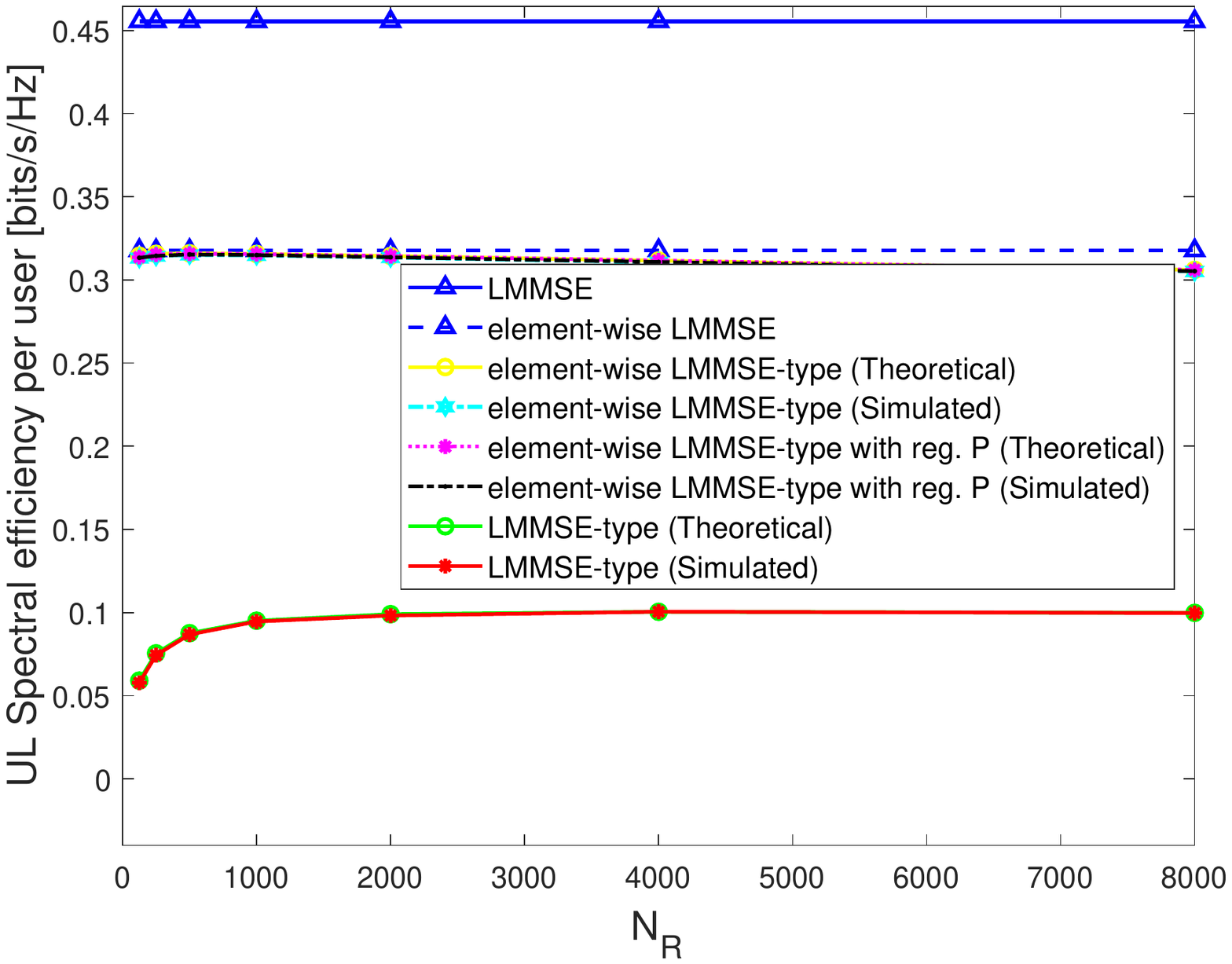}
			\caption{SE vs $N_R$ with $N_Q = 125$.}
			\label{fig: UL SE 125}
		\end{subfigure}
		\hspace{0.3cm}
		\begin{subfigure}{.5\textwidth}
			\centering
			\includegraphics[width=0.8\textwidth,trim={3cm 7cm 3cm 8cm}]{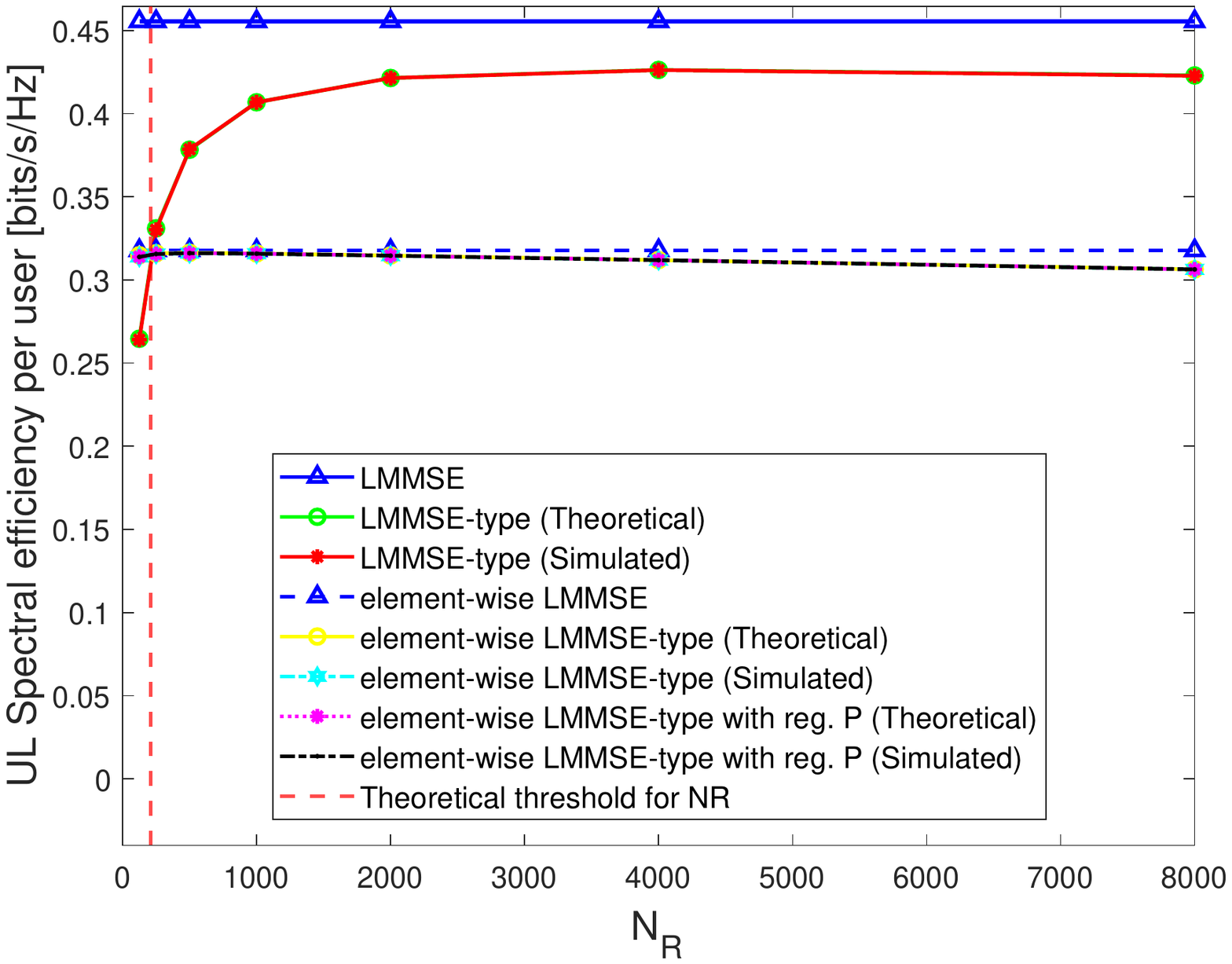}
			\caption{SE vs $N_R$ with $N_Q = 4000$.}
			\label{fig: UL SE 4000}
		\end{subfigure}
		\caption{UL SE for different channel estimation estimation techniques. In both the subplots, reg.~P stands for regularized $\elementwiseQestimate_{ju}$. \vspace{-6mm}}
		\label{fig: UL SE}
	\end{figure}
	
	In Fig. \ref{fig: UL SE}, it can be noticed that the theoretical SE, corresponding to LMMSE-type channel estimation, initially rises with $\pilotLengthForR$ to approach the SE that corresponds to LMMSE channel estimation, followed by a drop in the theoretical SE at $\pilotLengthForR= 8000$. In contrast, the theoretical SE, corresponding to element-wise LMMSE-type channel estimation (with and without regularized $\elementwiseQestimate_{ju}$), approaches the SE that corresponds to element-wise LMMSE channel estimation for $\pilotLengthForR$ value as low as $125$ and reaches its maximum at $\pilotLengthForR = 500$. Then, the theoretical SE reduces linearly with further increase in $\pilotLengthForR$ as the pilot overhead increases. Moreover, the simulated SEs match the theoretical values for all the three channel estimation techniques tested, thereby validating the derivations presented in the paper. Additionally, it can be observed that the regularization in estimating $\elementwiseQestimate_{ju}$ does not improve the SE significantly.
	
	The initial raise of the theoretical SEs is due to the improvement in the covariance estimates caused by the increase in the number of samples for estimation. However, a further increase in $\pilotLengthForR$ results in a drop of UL SEs due to the pre-log factor. Despite the improvement in estimation quality of the covariance matrices, the SEs drops because of the consumption of UL resources for the additional CovEst pilots. This validates the theoretical analysis done in Section \ref{sec:discussion}. 
	
	It can be seen from Fig.~\ref{fig: UL SE}(\subref{fig: UL SE 125}) and Fig.~\ref{fig: UL SE}(\subref{fig: UL SE 4000}) that, using element-wise LMMSE channel estimation instead of LMMSE channel estimation leads to a drop in SE. However, it is evident that the element-wise LMMSE-type channel estimation completely outperforms the LMMSE-type channel estimation for all the $\pilotLengthForR$ values, and for $\pilotLengthForQ = 125$. It can also be noted that even for $\pilotLengthForQ = 4000$, the element-wise LMMSE-type channel estimation outperforms the LMMSE-type channel estimation for $\pilotLengthForR = 125$. Moreover, for $\pilotLengthForQ = 4000$,  $\pilotLengthThreshold$ given in Section~\ref{sec:discussion} matches exactly with the $\pilotLengthForR$ value for which the LMMSE-type and element-wise LMMSE-type channel estimations have same performance. Therefore, the minimum SE guaranteed for a massive MIMO system with imperfect covariance information is the SE provided by the element-wise LMMSE channel estimator \footnote{Note that the objective is to have  $\pilotLengthForR$ and $\pilotLengthForQ$ as low as possible for guaranteeing a desired SE}. This SE can be achieved by using element-wise LMMSE-type channel estimation with very low values of $\pilotLengthForR$ and $\pilotLengthForQ$, and with low computational complexity. Finally, from simulations we compute the threshold value for $\pilotLengthForQ$ to be $299$, such that for $\pilotLengthForQ < 299$, element-wise LMMSE-type channel estimation always outperforms LMMSE-type channel estimation.
	\vspace{-3mm}
	\subsection{Downlink Spectral Efficiency}
	Similar to the UL simulation, in this simulation example, we consider the DL SE expressions that correspond to the three channel estimation techniques: LMMSE-type channel estimation, the element-wise LMMSE-type channel estimation, and the element-wise LMMSE-type channel estimation with regularized $\elementwiseQestimate_{ju}$. In Fig.~\ref{fig: DL SE}, we plot the SE as a function of $\pilotLengthForR$ for the three aforementioned channel estimation techniques. Fig.~\ref{fig: DL SE}(\subref{fig: DL SE 125}) depicts the SE values for $\pilotLengthForQ = 125$ and Fig.~\ref{fig: DL SE}(\subref{fig: DL SE 4000}) shows SE values for $\pilotLengthForQ = 4000$. We perform a study on these plots similar to the study done in \ref{subsec: uplink simulations}.
	
	\begin{figure}
		\begin{subfigure}{.5\textwidth}
			\centering
			\includegraphics[width=0.8\textwidth,trim={3cm 7cm 3cm 8cm}]{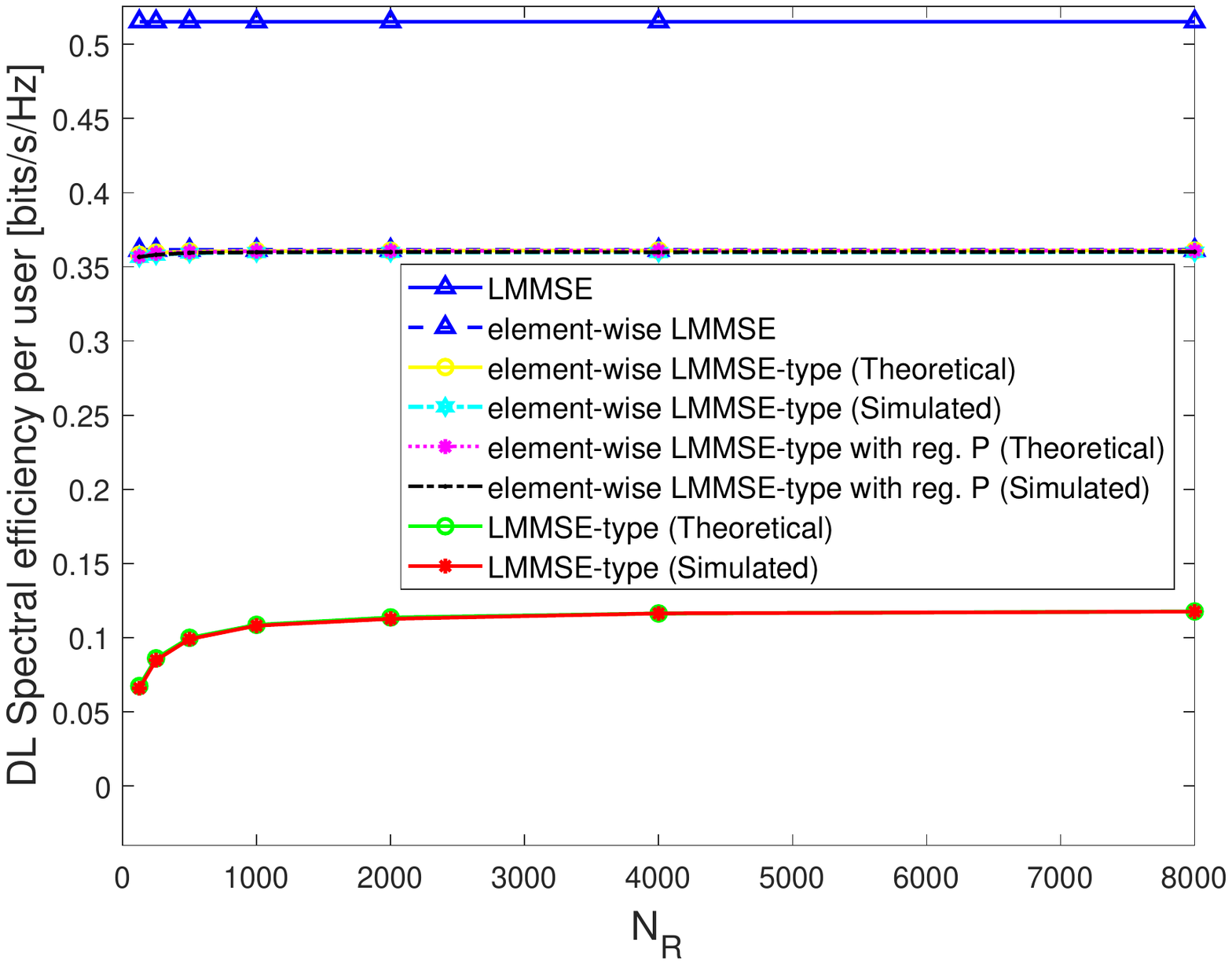}
			\caption{SE vs $N_R$ with $N_Q = 125$.}
			\label{fig: DL SE 125}
		\end{subfigure}
		\hspace{0.3cm}
		\begin{subfigure}{.5\textwidth}
			\centering
			\includegraphics[width=0.8\textwidth,trim={3cm 7cm 3cm 8cm}]{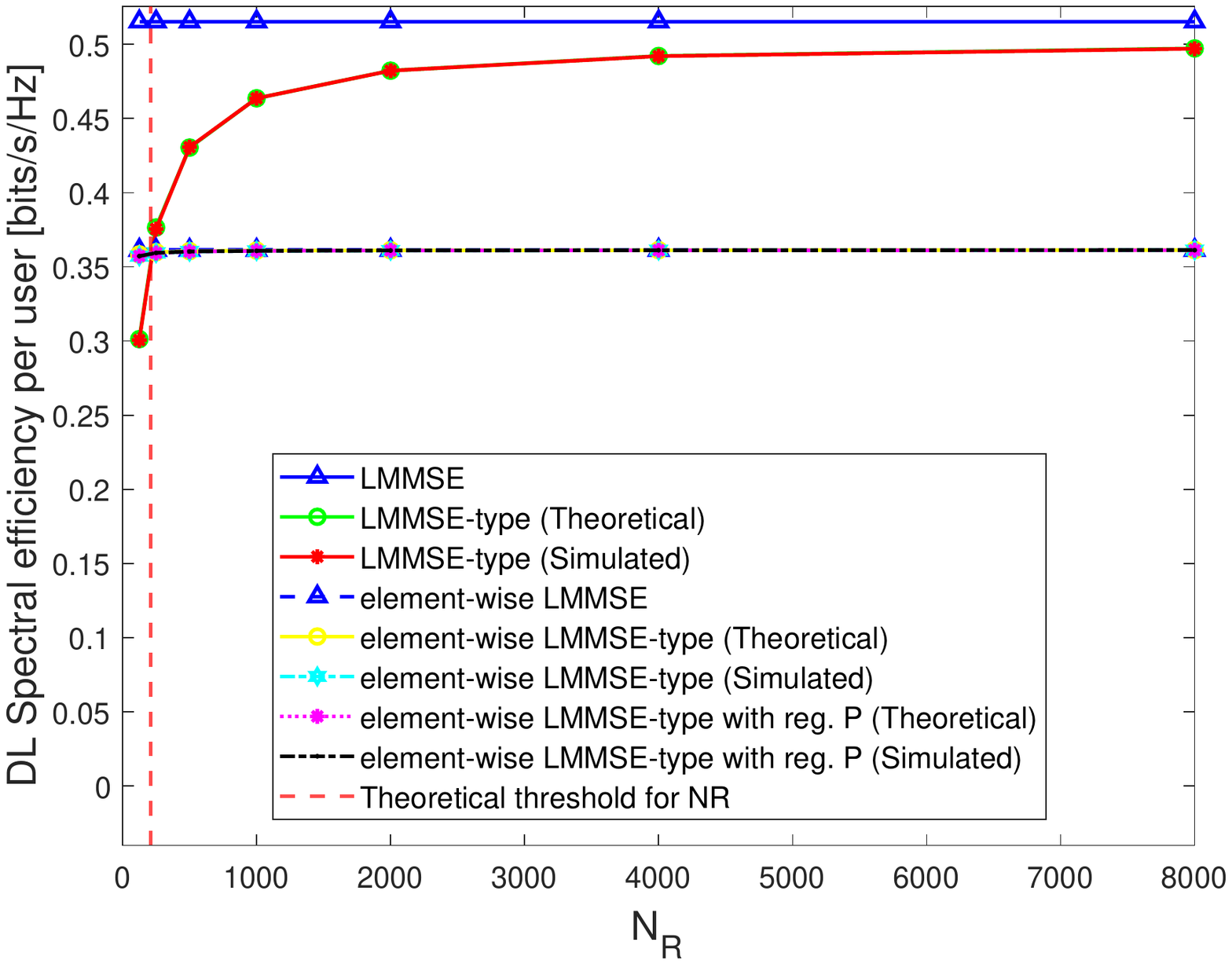}
			\caption{SE vs $N_R$ with $N_Q = 4000$.}
			\label{fig: DL SE 4000}
		\end{subfigure}
		\caption{DL SE for different channel estimation estimation techniques. In both the subplots, reg.~P stands for regularized $\elementwiseQestimate_{ju}$. \vspace{-6mm}}
		\label{fig: DL SE}
	\end{figure}
	
	In Fig. \ref{fig: DL SE}, it can be observed that the DL SE plots are similar to the plots in \ref{subsec: uplink simulations}. However, unlike in UL SE, an increase in $\pilotLengthForR$ does not result in a drop in SE as there is no pilot overhead in DL. The simulated SEs match the theoretical values for all the three channel estimation techniques, thereby validating the derivations presented in the paper. Moreover, for $\pilotLengthForQ = 4000$, $\pilotLengthThreshold$ given in Section~\ref{sec:discussion} matches exactly with the $\pilotLengthForR$ value for which the LMMSE-type and element-wise LMMSE-type channel estimations have same performance. From Fig.~\ref{fig: DL SE}(\subref{fig: DL SE 125}) and Fig.~\ref{fig: DL SE}(\subref{fig: DL SE 4000}), the minimum DL SE guaranteed for a massive MIMO system with imperfect covariance information is the SE provided by the element-wise LMMSE channel estimator. This SE can be achieved by using element-wise LMMSE-type channel estimation with very low values of $\pilotLengthForR$ and $\pilotLengthForQ$, and with low computational complexity. Finally, from simulations we compute the threshold value for $\pilotLengthForQ$ to be $300$, such that for $\pilotLengthForQ < 300$, element-wise LMMSE-type channel estimation always outperforms LMMSE-type channel estimation.
	
	\section{Conclusion}
	We have derived closed-form expressions for UL and DL SEs of a massive MIMO system which implements matched filter receiver and transmit combiners, respectively, as a function of $\pilotLengthForR$ and $\pilotLengthForQ$ which represent the UL pilot overhead. These combiners use channel estimates that utilize estimated covariance matrices in addition to channel observations. We have derived closed-form SE expressions for the case of LMMSE-type and element-wise LMMSE-type channel estimates (with and without regularization for $\elementwiseQestimate_{ju}$). Using theoretical analysis of these closed-form expressions and the using simulation results, we have demonstrated the impact of different values of $\pilotLengthForR$ and $\pilotLengthForQ$ on SEs of a user in a massive MIMO system, thereby presenting the closed-form expressions as the tools for solving the problem of choosing these parameters optimally. The derived theoretical SE expressions have been compared with the simulated SE values and an accurate agreement between them has also been demonstrated. Finally, using simulation results, we have shown that element-wise LMMSE-type channel estimation with very low values of $\pilotLengthForR$ and $\pilotLengthForQ$ provides the minimum guarantee SE, with low computational complexity.

	\appendices
	\section{Proof of Lemma~\ref{gaussianVectors}}
	\label{Appendix lemma 1}
	Let us start with a proof of \eqref{hAh}. Let the rank of the covariance matrix of $\mathbf{h}$, $\mathbf{R}$, be $K$. Then, we denote $\mathbf{\Lambda} \in \mathbb{R}^{K \times K}$ is a diagonal matrix containing positive eigenvalues of $\mathbf{R}$ and $\mathbf{U}\in \mathbb{R}^{M \times K}$ is a matrix containing $K$ eigenvectors corresponding to eigenvalues. Now, let us also define $\mathbf{B} \triangleq \mathbf{U}\mathbf{\Lambda}^{1/2} \in \mathbb{C}^{M \times K}$. Then, there exists a unique $\mathbf{g} \in \mathbb{C}^{K}$ such that $\mathbf{h} = \mathbf{B}\mathbf{g}$ and $\mathbb{E}\{\mathbf{g}\mathbf{g}^H\} = \mathbf{I}$. Therefore, we have $\mathbb{E}\{\mathbf{h}\mathbf{h}^H\arbSqrMatrix\mathbf{h}\mathbf{h}^H\} = \mathbf{B}\mathbb{E}\{\mathbf{g}\mathbf{g}^H\tilde{\arbSqrMatrix}\mathbf{g}\mathbf{g}^H\}\mathbf{B}^{H}$ where
	$\tilde{\arbSqrMatrix} \triangleq \mathbf{B}^{H}\arbSqrMatrix\mathbf{B}$.
	However, since $\mathbf{g}$ is distributed as $\mathcal{CN}(\mathbf{0},\mathbf{I})$, the term $\mathbb{E}\{\mathbf{g}\mathbf{g}^H\tilde{\arbSqrMatrix} \mathbf{g}\mathbf{g}^H\}$ can be evaluated as follows
	\begin{align*}
	\mathbb{E}\{[\mathbf{g}\mathbf{g}^H\tilde{\arbSqrMatrix} \mathbf{g}\mathbf{g}^H]_{ij}\}
	&= \sum_{p = 1}^{K}\sum_{q = 1}^{K} \mathbb{E}\{[\mathbf{g}]_i [\mathbf{g}]_p^{*} [\mathbf{g}]_q [\mathbf{g}]_j^{*}\} [\tilde{\arbSqrMatrix}]_{pq}
	= \begin{cases}
	[\tilde{\arbSqrMatrix}]_{ij}												&\text{if $i \neq j$} \\
	[\tilde{\arbSqrMatrix}]_{ii} + \trace(\tilde{\arbSqrMatrix})   					&\text{otherwise}
	\end{cases}
	\end{align*}
	and $\mathbb{E}\{\mathbf{g}\mathbf{g}^H\tilde{\arbSqrMatrix} \mathbf{g}\mathbf{g}^H\} = \tilde{\arbSqrMatrix} + \mathbf{I}\trace(\tilde{\arbSqrMatrix}).$
	Therefore, $\mathbb{E}\{\mathbf{h}\mathbf{h}^H\arbSqrMatrix\mathbf{h}\mathbf{h}^H\} = \mathbf{R}\arbSqrMatrix\mathbf{R} + \mathbf{R}\trace(\arbSqrMatrix\mathbf{R})$.
	
	Proof of \eqref{abshAh} is as follows. We first compute that
	\begin{align*}
	&\mathbb{E}\{|\mathbf{h}^H\arbSqrMatrix\mathbf{h}|^2\} = \mathbb{E}\{\mathbf{h}^H\arbSqrMatrix\mathbf{h}\mathbf{h}^H\arbSqrMatrix^H\mathbf{h}\}
	= \mathbb{E}\{\trace(\arbSqrMatrix\mathbf{h}\mathbf{h}^H\arbSqrMatrix^H\mathbf{h}\mathbf{h}^H)\}.
	\end{align*}
	Using \eqref{hAh}, we have $\mathbb{E}\{|\mathbf{h}^H\arbSqrMatrix\mathbf{h}|^2\} = |\trace(\arbSqrMatrix^H\mathbf{R})|^2 +\trace(\arbSqrMatrix\mathbf{R}\arbSqrMatrix^H\mathbf{R})$.
	\section{Proof of Lemma~\ref{Lemma:IWishproperties}}
	\label{Appendix lemma 2}
	Proof of \eqref{Wishart1} and \eqref{Wishart2} is given in \cite{604712}.
	
	Using the eigenvalue decomposition of $\arbSymMatrix = \mathbf{U} \mathbf{\Lambda} \mathbf{U}^{H}$, let us define $\tilde{\arbWishart} \triangleq \mathbf{U}^H\arbWishart\mathbf{U}$.
	It should be noted that $\tilde{\arbWishart}$ is distributed as $\mathcal{W}(N,\mathbf{I})$.
	Then, \eqref{Wishart3} can be proved as follows. First, we compute the following expectation term.
	\begin{align*}
	&\mathbb{E}\{\trace(\arbWishart^{-2}\arbSymMatrix)\} = \mathbb{E}\{\trace(\tilde{\arbWishart}^{-2}\mathbf{\Lambda})\}
	= \sum_{i=1}^{M} [\mathbb{E}\{\tilde{\arbWishart}^{-2}\}]_{ii}[\mathbf{\Lambda}]_{ii}
	\end{align*}
	But from \eqref{Wishart2}, we have
	\begin{align*}
	\mathbb{E}\{\trace(\arbWishart^{-2}\arbSymMatrix)\}
	&= \sum_{i=1}^{M} \frac{N}{(N-M)^3-(N-M)}[\mathbf{\Lambda}]_{ii} =\frac{N}{(N-M)^3-(N-M)}\trace(\arbSymMatrix)
	\end{align*}
	
	For \eqref{Wishart3}, we expand $\mathbb{E}\{|\trace(\arbWishart^{-1}\arbSqrMatrix)|^2\}$ using \eqref{Wishart2} as follows.
	\begin{align*}
	&\mathbb{E}\{|\trace(\arbWishart^{-1}\arbSqrMatrix)|^2\}
	= \sum_{p=1}^{M}\sum_{q=1}^{M}\sum_{r=1}^{M}\sum_{s=1}^{M} \mathbb{E}\{[\arbWishart^{-1}]_{pp}[\arbWishart^{-1}]_{ss}\}[\arbSqrMatrix]_{pp}[\arbSqrMatrix^{H}]_{ss}\\
	&= \sum_{p=1}^{M} \mathbb{E}\{[\arbWishart^{-1}]_{pp}[\arbWishart^{-1}]_{pp}\}[\arbSqrMatrix]_{pp}[\arbSqrMatrix^{H}]_{pp}
	+ \sum_{p=1}^{M}\sum_{s=1,s \neq p}^{M} \mathbb{E}\{[\arbWishart^{-1}]_{pp}[\arbWishart^{-1}]_{ss}\}[\arbSqrMatrix]_{pp}[\arbSqrMatrix^{H}]_{ss}\\
	&+ \sum_{p=1}^{M}\sum_{s=1,s \neq p}^{M} \mathbb{E}\{[\arbWishart^{-1}]_{ps}[\arbWishart^{-1}]_{sp}\}[\arbSqrMatrix]_{sp}[\arbSqrMatrix^{H}]_{ps}
	\end{align*}
	Using \eqref{Wishart2}, the above equation can be further simplified to \eqref{Wishart4}.
	\section{Proof of Lemma~\ref{den1R}}
	\label{Appendix lemma 3}
	Let us define a pair of mutually independent random vectors as follows.
	\begin{align*}
	&\contamination^{(1)}_{jju}[n] \triangleq \lmmseChannel^{(1)}_{jju}[n] - \mathbf{h}_{jju}, \; \contamination^{(2)}_{jju}[n] \triangleq \lmmseChannel^{(2)}_{jju}[n] - \mathbf{h}_{jju}
	\end{align*}
	The covariance matrices for $\contamination^{(1)}[n]$ and $\contamination^{(2)}[n]$ are identically equal to $\mathbf{Q}_{ju} - \mathbf{R}_{jju}$.
	Additionally, we also define mutually independent set of matrices
	\begin{align*}
	\singleObservationR_{jju}[n] \triangleq \lmmseChannel^{(1)}_{jju}[n](\lmmseChannel^{(2)}_{jju}[n])^{H} + \lmmseChannel^{(2)}_{jju}[n](\lmmseChannel^{(1)}_{jju}[n])^{H}
	\end{align*}
	for all $n \in \{1 \dots \pilotLengthForR\}$	such that $\rawRestimate_{jju} = \frac{1}{\pilotLengthForR} \sum_{n=1}^{N} \singleObservationR_{jju}[n]$
	by definition (i.e., \eqref{rawR}).
	
	Using the definition of $\contamination^{(1)}_{jju}[n]$ and $\contamination^{(2)}_{jju}[n]$, and also using Lemma \ref{gaussianVectors}, it can be shown that, for all $n = \{1\dots \pilotLengthForR\}$, we have
	\begin{align}
	&\mathbb{E}\{\singleObservationR_{jju}[n]\arbSqrMatrix\singleObservationR_{jju}[n]\}
	= \mathbf{R}_{jju}\arbSqrMatrix\mathbf{R}_{jju} +\frac{1}{2} \mathbf{Q}_{ju}\trace(\arbSqrMatrix \mathbf{Q}_{ju})
	+ \frac{1}{2} \mathbf{R}_{jju}\trace(\arbSqrMatrix \mathbf{R}_{jju}) \label{RAR_interm}\\
	&\mathbb{E}\{|\trace(\singleObservationR_{jju}[n]\arbSqrMatrix)|^2\} = |\trace(\mathbf{R}_{jju}\arbSqrMatrix)|^2+ \frac{1}{2}\trace(\arbSqrMatrix\mathbf{Q}_{ju}\arbSqrMatrix^H\mathbf{Q}_{ju}) +\frac{1}{2}\trace(\arbSqrMatrix\mathbf{R}_{jju}\arbSqrMatrix^H\mathbf{R}_{jju}). \label{absTraceRA_interm}
	\end{align}
	Finally, along with the equation $\rawRestimate_{jju} = \frac{1}{\pilotLengthForR} \sum_{n=1}^{N} \singleObservationR_{jju}[n]$, \eqref{RAR_interm} and \eqref{absTraceRA_interm} will result in \eqref{RAR} and \eqref{absTraceRA}, respectively.
	\vspace{-6mm}
	\section{Proof of Lemma~\ref{el_IWishproperties}}
	\label{Appendix lemma 5}
	Since $\arbDiagonalWishart = \mathbf{Z}/2$, and the elements of the diagonal matrix $\mathbf{Z}$ are $\chi^2$ distributed with $2N$ degrees of freedom, we have $\mathbb{E}\{[\arbDiagonalWishart^{-1}]_{pp}\} = 2\mathbb{E}\{[\mathbf{Z}^{-1}]_{pp}\} = 1/(N-1)$ and $\mathbb{E}\{[\arbDiagonalWishart^{-1}]^2_{pp}\} = 4\mathbb{E}\{[\mathbf{Z}^{-1}]^2_{pp}\} = 1/(N-1)(N-2)$.
	
	Using the above results, \eqref{YA1YA2} can be derived as follows
	\begin{align*}
	\mathbb{E}\{\trace(\arbDiagonalWishart^{-1} \arbSqrMatrix_1 \arbDiagonalWishart^{-1} \arbSqrMatrix_2)\}
	&= \bigg(\frac{1}{N-1}\bigg)^2\sum_{p=1}^{M}\sum_{q\ne p} [\arbSqrMatrix_1]_{pq}[\arbSqrMatrix_2]_{qp} + \frac{1}{(N-1)(N-2)} \sum_{p=1}^{M}[\arbSqrMatrix_1]_{pp}[\arbSqrMatrix_2]_{pp}\\
	&= \genericConstantOne_1 \trace(\arbSqrMatrix_1\arbSqrMatrix_2) + \genericConstantOne_2 \trace(\arbSqrMatrix_{1d}\arbSqrMatrix_{2d})
	\end{align*}
	where $\genericConstantOne_1 \triangleq 1/(N-1)^2$, $\genericConstantOne_2 \triangleq 1/((N-1)^2(N-2))$, $\arbSqrMatrix_{1d} \triangleq \diag(\arbSqrMatrix_1)$, and $\arbSqrMatrix_{2d} \triangleq \diag(\arbSqrMatrix_2)$.
	
	In what follows, proof of \eqref{abstrace_YA} is presented
	\begin{align*}
	\mathbb{E}\{|\trace(\arbDiagonalWishart^{-1}\arbSqrMatrix)|^2\} &= \frac{1}{(N-1)^2}\sum_{p=1}^{M}\sum_{q\ne p} [\arbSqrMatrix]_{pp}[\arbSqrMatrix]^{*}_{qq} + \frac{1}{(N-1)(N-2)}\sum_{p=1}^{M}\abs{[\arbSqrMatrix]_{pp}}^2 \nonumber \\
	&= \genericConstantOne_1 |\trace(\arbSqrMatrix)|^2 + \genericConstantOne_2\trace(\arbSqrMatrix^H_d \arbSqrMatrix_d)
	\end{align*}
	where $\arbSqrMatrix_d \triangleq \diag(\arbSqrMatrix)$.
	\vspace{-6mm}
	\section{Proof of Lemma \ref{el_den1R}}
	\label{Appendix lemma 6}
	Let us define a pair of mutually independent random vectors as follows.
	\begin{align*}
	&\contamination^{(1)}_{jju}[n] \triangleq  \lmmseChannel^{(1)}_{jju}[n] - \mathbf{h}_{jju},\;
	\contamination^{(2)}_{jju}[n] \triangleq  \lmmseChannel^{(2)}_{jju}[n] - \mathbf{h}_{jju}
	\end{align*}
	The covariance matrices for $\contamination^{(1)}_{jju}[n]$ and $\contamination^{(2)}_{jju}[n]$ are identically equal to 
	$\mathbf{Q}_{ju} - \mathbf{R}_{jju}$. Additionally, we also define mutually independent set of matrices as
	\begin{align*}
	\elementWiseSingleObservationR_{jjk}[n] \triangleq \diag(\lmmseChannel^{(1)}_{jju}[n](\lmmseChannel^{(2)}_{jju}[n])^{H} + \lmmseChannel^{(2)}_{jju}[n](\lmmseChannel^{(1)}_{jju}[n])^{H})
	\end{align*}
	for all $n \in \{1\dots \pilotLengthForR\}$ such that $\elementwiseRestimate_{jju} = \frac{1}{N} \sum_{n=1}^{N} \elementWiseSingleObservationR_{jju}[n]$
	by definition (i.e., \eqref{rawS}).
	
	Using the definitions of $\contamination^{(1)}_{jju}[n]$ and $\contamination^{(1)}_{jju}[n]$ together with Lemma~\ref{gaussianVectors} (for scalar case), and Lemma~\ref{h1h2},
	it can be shown that
	\begin{align*}
	\mathbb{E}\{[\elementWiseSingleObservationR_{jju}]_{pp}[\elementWiseSingleObservationR_{jju}]_{qq}\}
	&= \mathbb{E}\{|[\mathbf{h}_{jju}]_p|^2|[\mathbf{h}_{jju}]_q|^2\}
	+ \frac{1}{2}[\mathbf{R}_{jju}]_{pq}[\mathbf{Q}_{ju}-\mathbf{R}_{jju}]_{qp}
	+ \frac{1}{2} [\mathbf{Q}_{ju}-\mathbf{R}_{jju}]_{pq}[\mathbf{R}_{jju}]_{qp}\nonumber \\
	&+ \frac{1}{2} [\mathbf{Q}_{ju}-\mathbf{R}_{jju}]_{pq}[\mathbf{Q}_{ju}-\mathbf{R}_{jju}]_{qp} \nonumber \\
	&= [\elementwiseR_{jju}]_{pp}[\elementwiseR_{jju}]_{qq} +\frac{1}{2} [\mathbf{Q}_{jju}]_{pq}[\mathbf{Q}_{jju}]_{pq}
	+\frac{1}{2} [\mathbf{R}_{jju}]_{pq}[\mathbf{R}_{jju}]_{pq}.
	\end{align*}
	Therefore, we have
	\begin{align}
	&\mathbb{E}\{[\elementWiseSingleObservationR_{jju}\arbSqrMatrix\elementWiseSingleObservationR_{jju}]_{pq}\}
	= [\arbSqrMatrix]_{pq}\big\{[\elementwiseR_{jju}]_{pp}[\elementwiseR_{jju}]_{qq}
	+ \frac{1}{2} [\mathbf{Q}_{jju}]_{pq}[\mathbf{Q}_{jju}]_{pq} + \frac{1}{2} [\mathbf{R}_{jju}]_{pq}[\mathbf{R}_{jju}]_{pq}\big\} \label{SAS_interm}\\
	&\mathbb{E}\{|\trace(\elementWiseSingleObservationR_{jju}\arbDiaMatrix)|^2\}
	= \sum_{p=1}^{M}\sum_{q=1}^{M} \bigg\{[\elementwiseR_{jju}]_{pp}[\elementwiseR_{jju}]_{qq}
	+ \frac{1}{2}[\mathbf{Q}_{ju}]_{pq}[\mathbf{Q}_{ju}]_{pq} + \frac{1}{2} [\mathbf{R}_{jju}]_{pq}[\mathbf{R}_{jju}]_{pq}\bigg\}[\arbDiaMatrix]_{pp}[\arbDiaMatrix]_{qq} \nonumber \\
	&= |\trace(\elementwiseR_{jju}\arbDiaMatrix)|^2 +\frac{1}{2} \sum_{p=1}^{M}\sum_{q=1}^{M} [\arbDiaMatrix(\mathbf{Q}_{ju}\circ \mathbf{Q}_{ju})\arbDiaMatrix]_{pq} +\frac{1}{2} \sum_{p=1}^{M}\sum_{q=1}^{M} [\arbDiaMatrix(\mathbf{R}_{jju}\circ \mathbf{R}_{jju})\arbDiaMatrix]_{pq}.\label{abstraceSA_interm}
	\end{align}
	
	Finally, along with the equation $\elementwiseRestimate_{jju} = \frac{1}{N} \sum_{n=1}^{N} \elementWiseSingleObservationR_{jju}[n]$, \eqref{SAS_interm} and \eqref{abstraceSA_interm} will result in
	\eqref{SAS} and \eqref{el_absTraceRA}, respectively.
	\vspace{-6mm}
	\section{Proof of Lemma~\ref{el_IWishpropertiesReg}}
	\label{Appendix lemma 7}
	Expression \eqref{YA1YA2Reg} is derived as follows:
	\begin{align*}
	\mathbb{E}\{\trace(\regElementwiseQestimate^{-1}_{ju} \arbSqrMatrix_1 \regElementwiseQestimate^{-1}_{ju} \arbSqrMatrix_2)\}\!
	&= \!\sum_{p=1}^{M}\!\sum_{q\ne p}\! \mathbb{E}\!\{\![\regElementwiseQestimate^{-1}]_{pp}\!\}\!\mathbb{E}\!\{\![\regElementwiseQestimate^{-1}]_{qq}\!\}\![\arbSqrMatrix_1]_{pq}[\arbSqrMatrix_2]_{qp}
	+ \sum_{p=1}^{M}\mathbb{E}\!\{\![\regElementwiseQestimate^{-1}]_{pp}^2\!\}\![\arbSqrMatrix_1]_{pp}[\arbSqrMatrix_2]_{pp}\\
	&=  \trace\left(\expectationMatrixFirstOrder\arbSqrMatrix_1\expectationMatrixFirstOrder\arbSqrMatrix_2\right) +  \trace\left((\expectationMatrixSecondOrder-\expectationMatrixFirstOrder^2)\arbSqrMatrix_{1d}\arbSqrMatrix_{2d}\right)
	\end{align*}
	where  $\arbSqrMatrix_{1d} \triangleq \diag(\arbSqrMatrix_1)$ and $\arbSqrMatrix_{2d} \triangleq \diag(\arbSqrMatrix_2)$. In what follows, proof of \eqref{abstrace_YAReg} is presented
	\begin{align*}
	\mathbb{E}\{|\trace(\regElementwiseQestimate^{-1}\arbSqrMatrix)|^2\} &= \sum_{p=1}^{M}\sum_{q\ne p} \mathbb{E}\{[\regElementwiseQestimate^{-1}]_{pp}\}\mathbb{E}\{[\regElementwiseQestimate^{-1}]_{qq}\}[\arbSqrMatrix]_{pp}[\arbSqrMatrix]^{*}_{qq} + \sum_{p=1}^{M} \mathbb{E}\{[\regElementwiseQestimate^{-1}]_{pp}^2\} \abs{[\arbSqrMatrix]_{pp}}^2\\
	&= |\trace(\expectationMatrixFirstOrder \arbSqrMatrix)|^2 + \trace\left((\expectationMatrixSecondOrder-\expectationMatrixFirstOrder^2)\arbSqrMatrix^H_d \arbSqrMatrix_d\right)
	\end{align*}
	where $\arbSqrMatrix_d \triangleq \diag(\arbSqrMatrix)$.

	\ifCLASSOPTIONcaptionsoff
	\newpage
	\fi
	
	\bibliographystyle{IEEEtran}
	\bibliography{./Template}


	
	%
	%
\end{document}